\documentclass[11pt]{article}
\usepackage{epsfig,graphics,caption,subcaption}
\usepackage{amssymb,amsmath,amsthm,bm,color,verbatim}

\usepackage[top=0.5in,bottom=0.66in,left=0.5in,right=0.5in]{geometry}

\newtheorem{theo}{Theorem}[section]
\newtheorem{lem}{Lemma}[section]
\newtheorem{pro}{Proposition}[section]
\newtheorem{cor}{Corollary}[section]

\newtheorem{example}{Example}[section]

\numberwithin{equation}{section}

\def\ZZ{\mathbb{Z}}
\def\RR{\mathbb{R}}
\def\NN{\mathbb{N}}
\def\Z{\mathbb{Z}}
\def\R{\mathbb{R}}
\def\N{\mathbb{N}}
\def\C{\mathbb{C}}

\def\bx{{\bm x}}

\def\b0{{\bf 0}}

\def\gb{\beta}

\def\go{\omega}

\def\gbo{{\boldsymbol \omega}}

\def\ol{\overline}
\def\tp{\top}

\def\wh{\widehat}

\newcommand{\bhanignore}[1]{} %use #1 to turn the deleted part back.

\newcommand{\lpf}{a} %low-pass filter
\newcommand{\hpf}{b} %high-pass filter

\newcommand{\AS}{\operatorname{AS}}
\newcommand{\bo}{\mathcal{O}}
\newcommand{\sr}{\operatorname{sr}}  %sum rules
\newcommand{\sm}{\operatorname{sm}}  %smoothness exponent
\newcommand{\vmo}{\operatorname{vm}} %\vm=vec{m}. Vanishing moment
\newcommand{\lpm}{\operatorname{lpm}}  %linear-phase moments
\newcommand{\gl}{\lambda}
\newcommand{\pv}{\mathsf{v}}
\newcommand{\sym}{\mathsf{S}}

\begin{document}
\date{}

\title{Symmetric Canonical Quincunx Tight Framelets with High Vanishing Moments and Smoothness
\footnote{Research of B. Han was supported by
the Natural Sciences and Engineering Research Council of Canada (NSERC Canada) under
Grant 05865.
%Research of Q. Jiang was supported by ???.
Research of Z. Shen was supported by several grants from Singapore.
Research of X. Zhuang was  supported by the Research Grants Council of Hong Kong (Project No. CityU 11304414).}
}

\author{Bin Han${}^a$, \;  Qingtang Jiang${}^b$, \; Zuowei
Shen${}^c$, \; and Xiaosheng Zhuang${}^d$\\
\\
${}^a${\it {\small  Department of Mathematical and Statistical Sciences}} \\
{\it {\small University of Alberta,  Edmonton, Alberta, Canada T6G 2G1}}\\
{\small bhan@ualberta.ca}\\
${}^b${\it {\small Department of Mathematics and Computer Science}}\\
{\it {\small University of Missouri--St. Louis, St. Louis, MO 63121, USA}}\\
{\small jiangq@umsl.edu}\\
${}^c${\it {\small Department of Mathematics, National University of Singapore}}\\
{\small {\it 10 Lower Kent Ridge Road, Singapore, 119076 }}\\
{\small matzuows@nus.edu.sg}\\
${}^d${\it {\small Department of Mathematics,
City University of Hong Kong}}\\
{\small {\it Tat Chee Avenue,  Kowloon Tong, Hong Kong}}\\
{\small xzhuang7@cityu.edu.hk }
}
 \maketitle

\begin{abstract}

In this paper, we propose an approach to construct a family of two-dimensional compactly supported real-valued \emph{symmetric} quincunx tight framelets $\{\phi; \psi_1,\psi_2,\psi_3\}$ in $L_2(\R^2)$ with  arbitrarily high orders of vanishing moments. Such symmetric quincunx tight framelets are associated with quincunx tight framelet filter banks $\{\lpf;\hpf_1,\hpf_2,\hpf_3\}$ having increasing orders of vanishing moments and enjoying the additional double canonical properties:
\[
\hpf_1(k_1,k_2)=(-1)^{1+k_1+k_2} \lpf(1-k_1,-k_2),\qquad \hpf_3(k_1,k_2)=(-1)^{1+k_1+k_2} \hpf_2(1-k_1,-k_2),\qquad \forall\, k_1,k_2\in \Z.
\]
Moreover, the supports of all the high-pass filters $\hpf_1, \hpf_2,\hpf_3$ are no larger than that of the low-pass filter $\lpf$. For a low-pass filter $\lpf$ which is not a quincunx orthonormal wavelet filter, we show that a quincunx tight framelet filter bank $\{\lpf;\hpf_1,\ldots,\hpf_L\}$ with $\hpf_1$ taking the above canonical form must have $L\ge 3$ high-pass filters. Thus, our family of symmetric double canonical quincunx tight  framelets has the minimum number of generators. Numerical calculation indicates that this family of symmetric double canonical quincunx tight framelets can be arbitrarily smooth. Using one-dimensional filters having linear-phase moments, in this paper we also provide a second approach to construct multiple  canonical quincunx tight  framelets with symmetry. In particular, the second approach yields a family of $6$-multiple canonical real-valued quincunx tight framelets in $L_2(\R^2)$ and a family of double canonical complex-valued quincunx tight framelets in $L_2(\R^2)$ such that both of them have symmetry and arbitrarily increasing orders of smoothness and vanishing moments. Several examples are provided to illustrate our general construction and theoretical results on canonical quincunx tight  framelets in $L_2(\R^2)$ with symmetry, high vanishing moments, and smoothness.
Symmetric quincunx tight framelets constructed by both approaches in this paper are of particular interest for their applications in computer graphics and image processing due to their polynomial preserving property, full symmetry, short support, and high smoothness and vanishing moments.

\bigskip



%\bigskip {\it AMS 2000 Math Subject Classification:}

\end{abstract}

\section{Introduction and motivations}
\label{sec:intro}

\setcounter{equation}{0}

%\subsection{Some background on tight framelets and tight framelet filter banks}

In this paper we study quincunx tight  framelets having full symmetry, short support, high vanishing moments and smoothness.
We say that a $d\times d$ matrix $M$ is \emph{a dilation matrix} if $M$ is an integer matrix having all its eigenvalues greater than one in modulus. In dimension two, typical and important dilation matrices $M$ include
\begin{equation}\label{quincunxmatrix}
2I_2:=\left[ \begin{matrix} 2 &0\\ 0 &2\end{matrix}\right],\qquad
M_{\sqrt{2}}:=\left[ \begin{matrix}
1&1\\ 1&-1 \end{matrix}
\right], \qquad N_{\sqrt{2}}:=\left[ \begin{matrix}
1&-1\\ 1&1 \end{matrix}
\right],
\end{equation}
where $M_{\sqrt{2}}$ and $N_{\sqrt{2}}$ are called \emph{quincunx dilation matrices}.
For functions $\phi, \psi_{1},\ldots,\psi_{L}$ in $L_2(\R^d)$, we say that $\{\phi; \psi_{1},\ldots, \psi_{L}\}$ is \emph{a tight $M$-framelet} for $L_2(\R^d)$ if the affine system
$\AS(\{\phi; \psi_{1},\ldots, \psi_{L}\})$ is a normalized tight frame of $L_2(\R^d)$; that is,
\begin{equation}\label{tf}
\|f\|_{L_2(\R^d)}^2=\sum_{k\in \Z^d} |\langle f, \phi(\cdot-k)\rangle|^2+\sum_{j=0}^\infty \sum_{\ell=1}^L \sum_{k\in \Z^d} |\langle f, |\det(M)|^{j/2} \psi_{\ell}(M^j\cdot-k)\rangle|^2,\qquad \forall\, f\in L_2(\R^d),
\end{equation}
where the affine system generated by the functions $\phi, \psi_{1},\ldots, \psi_{L}$ is defined to be
\[
\AS(\{\phi; \psi_{1},\ldots, \psi_{L}\}):=
\{\phi(\cdot-k)\; : \; k\in \Z^d\}\cup\{ |\det(M)|^{j/2} \psi_{\ell}(M^j\cdot-k) \; : \;
j\in \N\cup\{0\}, k\in \Z^d, \ell=1,\ldots, L\},
\]
$\langle f,g\rangle:=\int_{\R^d} f(x)\ol{g(x)}dx$ is the  inner product, and $\|f\|_{L_2(\R^d)}:=\sqrt{\langle f, f\rangle}$ is the $L_2$-norm.
If $\AS(\{\phi; \psi_{1},\ldots, \psi_{L}\})$
is an orthonormal basis of $L_2(\R^d)$, then $\{\phi; \psi_{1},\ldots, \psi_{L}\}$ is called \emph{an orthonormal $M$-wavelet}. It is known in \cite[Proposition~4]{Han:acha:2012} that if $\AS(\{\phi;\psi_1,\ldots,\psi_L\})$ is a normalized tight frame (or an orthonormal basis) for $L_2(\R^d)$, then the homogeneous affine system $\AS(\{\psi_1,\ldots,\psi_L\})$ must be a normalized tight frame (or an orthonormal basis) for $L_2(\R^d)$ as well, where
\begin{equation} \label{has}
\AS(\{\psi_{1},\ldots, \psi_{L}\}):=\{ |\det(M)|^{j/2} \psi_{\ell}(M^j\cdot-k) \; : \;
j\in \Z, k\in \Z^d, \ell=1,\ldots, L\}.
\end{equation}
Tight $M$-framelets and orthonormal $M$-wavelets are often derived from $M$-refinable functions. By $l_0(\Z^d)$ we denote the set of all finitely supported sequences $u=\{u(k)\}_{k\in \Z^d}$ on $\Z^d$. For $u\in l_0(\Z^d)$, its Fourier series (or symbol) $\wh{u}$ is a $2\pi\Z^d$-periodic trigonometric polynomial defined by $\wh{u}(\gbo):=\sum_{k\in \Z^d} u(k) e^{-ik\cdot\gbo}$, $\gbo\in \R^d$. For $\lpf,\hpf_{1},\ldots, \hpf_{L}\in l_0(\Z^d)$ such that $\wh{\lpf}(0)=\sum_{k\in \Z^d} \lpf(k)=1$, the following functions
\begin{equation}\label{phi:psi}
\wh{\phi}(\gbo):=\prod_{j=1}^\infty \wh{\lpf}((M^\tp)^{-j}\gbo),\qquad
\wh{\psi_{\ell}}(\gbo):=\wh{\hpf_\ell}((M^\tp)^{-1}\gbo) \wh{\phi}((M^\tp)^{-1}\gbo),\qquad \gbo\in \R^d, \ell=1,\ldots, L
\end{equation}
are well defined (\cite{Dau}). In the spatial domain, $\phi$ satisfies the following \emph{refinement equation}
\[
\phi=|\det(M)|\sum_{k\in \Z^d} \lpf(k) \phi(M\cdot-k)
\]
and $\phi$ is called \emph{the %standard
$M$-refinable function/distribution} associated with the \emph{filter/mask} $\lpf$.
For the functions $\phi$, $\psi_1$, $\ldots$, $\psi_L$ defined in \eqref{phi:psi} through the filters $\lpf,\hpf_1,\ldots,\hpf_L\in l_0(\Z^d)$ satisfying $\wh{\lpf}(0)=1$,
$\{\phi; \psi_{1},\ldots, \psi_{L}\}$ is a tight $M$-framelet for $L_2(\R^d)$ if and only if $\{\lpf; \hpf_{1},\ldots, \hpf_{L}\}$ is \emph{a tight $M$-framelet filter bank}; that is,
\begin{equation}\label{tffb}
|\wh{\lpf}(\gbo)|^2+\sum_{\ell=1}^L |\wh{\hpf_{\ell}}(\gbo)|^2=1
\quad \mbox{and}\quad
\ol{\wh{\lpf}(\gbo)}\wh{\lpf}(\gbo+2\pi \xi)+
\sum_{\ell=1}^L \ol{\wh{\hpf_{\ell}}(\gbo)}\wh{\hpf_{\ell}}(\gbo+2\pi \xi)=0,\qquad \xi\in \Omega_{M}\backslash \{0\},
\end{equation}
where $\Omega_{M}$ is a set of representatives of the distinct cosets of the quotient group $[(M^\tp)^{-1}\Z^d]/\Z^d$ and is given by
\begin{equation}\label{Omega}
\Omega_{M}:=[(M^\tp)^{-1} \Z^d] \cap [0,1)^d.
\end{equation}
As observed in \cite{Han:laa:2002,Han:jcam:2003}, the equations in \eqref{tffb} for a tight $M$-framelet filter bank only depend on the lattice $M\Z^d$ instead of $M$ itself. That is, for two $d\times d$ integer matrices $M$ and $N$ satisfying
\begin{equation}\label{same:dilation}
M\Z^d=N\Z^d,
\end{equation}
$\{\lpf;\hpf_1,\ldots,\hpf_L\}$ is a tight $M$-framelet filter bank if and only if it is a tight $N$-framelet filter bank.
This simple observation in \cite{Han:laa:2002,Han:jcam:2003} comes from the fact that \eqref{same:dilation} is equivalent to $M=NE$ for some integer matrix $E$ with $|\det(E)|=1$, which trivially implies $(M^\tp)^{-1}\Z^d=(N^\tp)^{-1}\Z^d$. For example,
the two quincunx dilation matrices in \eqref{quincunxmatrix} satisfy $M_{\sqrt{2}}\Z^2=N_{\sqrt{2}}\Z^2$, which is the quincunx lattice $\{ (j,k)\in \Z^2 \, : \, j+k\; \mbox{is even}\}$.

When \eqref{tffb} holds, it was proved in \cite{RS1} that the corresponding homogeneous affine system $\AS(\{\psi_1,\ldots,\psi_L\})$ forms a normalized tight frame in $L_2(\R^d)$, which is called the unitary extension principle.
%The equations in \eqref{tffb} are called \emph{the unitary extension principle} in \cite{RS1}.
Under various conditions on $\phi,\psi_1,\ldots,\psi_L$ and $\lpf,\hpf_1,\ldots,\hpf_L$, tight framelets have been studied in \cite{CHS_2002,Daubechies2003,Han_1997,RS1} and references therein.
Under the natural and necessary condition $\wh{\lpf}(0)=1$,
the above one-to-one correspondence between a tight $M$-framelet $\{\phi; \psi_1,\ldots,\psi_L\}$ and a tight $M$-framelet filter bank $\{a;b_1,\ldots,b_L\}$ has been presented in \cite[Lemma~2.1, Theorems~2.2 and 2.3]{Han:jcam:2003} or more generally, \cite[Corollary~12 and Theorem~17]{Han:acha:2012} for fully nonstationary tight framelets. In particular, if $\{\lpf;\hpf_1,\ldots,\hpf_L\}$ is a tight $M$-framelet filter bank with $\wh{\lpf}(0)=1$, then the functions $\phi,\psi_1,\ldots,\psi_L$ defined in \eqref{phi:psi} must be square integrable functions in $L_2(\R^d)$ (see \cite[Lemma~2.1]{Han:jcam:2003}).
Due to this one-to-one correspondence between tight $M$-framelets and tight $M$-framelet filter banks, in this paper we shall concentrate on tight $M$-framelet filter banks. Wavelets and framelets using the quincunx dilation matrices in \eqref{quincunxmatrix} are called quincunx wavelets or quincunx framelets in this paper.

For some applications such as computer graphics and computer aided geometric design,
symmetry of framelets and wavelets is highly desired.
Let us now discuss the general symmetry of a filter.
Let $G$ be a finite set of $d\times d$ integer matrices that forms a group under the usual matrix multiplication. We say that a filter $\lpf\in l_0(\Z^d)$ is \emph{$G$-symmetric about a point  $\mathbf{c}\in\RR^d$} if
\begin{equation}
\label{def:mask-sym-group-time}
\lpf(E(k-\mathbf{c})+\mathbf{c}) = \lpf(k),\qquad\forall k\in \ZZ^d\; \mbox{ and }\; \forall E\in G.
\end{equation}
%
%A real-valued filter $\lpf$ is symmetric about a point $\mathbf{c}$ if and only if $\lpf$ is $\{I_2,-I_2\}$-symmetric about $\mathbf{c}$.
However, the symmetry of a low-pass filter $\lpf$ does not automatically guarantee the symmetry of the $M$-refinable function $\phi$ defined in \eqref{phi:psi}. As discussed in \cite{Han:laa:2002,Han:simaa:2003,Han:acha:2004}, some compatibility condition is needed.
We say that a dilation matrix $M$ is \emph{compatible} with a symmetry group $G$ if $MEM^{-1}\in G$ for all $E\in G$.
If $M$ is compatible with a symmetry group $G$,
then $\phi$ in \eqref{phi:psi} is $G$-symmetric about $c_\phi:=(M-I_d)^{-1}\mathbf{c}$ (i.e.,
$\phi(E(\cdot-c_\phi)+c_\phi)=\phi$ for all $E\in G$) if and only if $\lpf$ is $G$-symmetric about $\mathbf{c}$ (see \cite[Proposition~2.1]{Han:acha:2004} and \cite{Han:laa:2002,Han:simaa:2003}). One of the commonly used two-dimensional symmetry groups in computer graphics is the dihedral group $D_4$ given by
\begin{equation}\label{def:D4}
D_4:=\left \{\pm\left[\begin{matrix}1& 0\\0 &1\end{matrix}\right],
\pm\left[\begin{matrix}1& 0\\0 &-1\end{matrix}\right],
\pm\left[\begin{matrix}0& 1\\1 &0\end{matrix}\right],
\pm\left[\begin{matrix}0& 1\\-1 &0\end{matrix}\right]\right\}.
\end{equation}
Note that $M_{\sqrt{2}}$ is compatible with the symmetry group $D_4$ and its subgroup $\{I_2,-I_2\}$, but it is not compatible with the symmetry group
$D_4^+:=\{\pm \mbox{diag}(1, 1), \pm \mbox{diag}(1,-1)\}$. A matrix $N$ is \emph{$G$-equivalent} to $M$ if $N=EMF$ for some $E,F\in G$. Note that $N_{\sqrt{2}}$ in \eqref{quincunxmatrix} is $D_4$-equivalent to $M_{\sqrt{2}}$. It is of interest to point out here that \cite[Theorem~2]{Han:stmalo:2003} shows that every $2\times 2$ matrix $M$ compatible with $D_4$ must be $D_4$-equivalent to either $M=cI_2$ or $M=cM_{\sqrt{2}}$ for some $c\in \Z$.
This makes the quincunx dilation matrices $M_{\sqrt{2}}$ and $N_{\sqrt{2}}$ particularly interesting for constructing tight framelets having the full symmetry $D_4$.
For a low-pass $D_4$-symmetric filter $a$, since $N_{\sqrt{2}}$ is $D_4$-equivalent to $M_{\sqrt{2}}$, we shall see in this paper that the $N_{\sqrt{2}}$-refinable function is just a shifted version of the $M_{\sqrt{2}}$-refinable function. However, the $M_{\sqrt{2}}$-refinable function and
the $N_{\sqrt{2}}$-refinable function associated with a low-pass filter $a$ without symmetry could be completely different (\cite{CD_1993,Han:laa:2002}). Because we are mainly interested in symmetric quincunx tight framelet filter banks, as a consequence, there are no essential differences for using either $M_{\sqrt{2}}$ or $N_{\sqrt{2}}$. Therefore, for simplicity, we mainly discuss the dilation matrix $M_{\sqrt{2}}$ in this paper.

A tight $M$-framelet filter bank $\{\lpf; \hpf_{1},\ldots, \hpf_{L}\}$ with $L=|\det(M)|-1$ is called \emph{an orthonormal $M$-wavelet filter bank}. It is a simple consequence of the equations in \eqref{tffb} (by rewriting the equations in \eqref{tffb} in a matrix form) that the low-pass filter $\lpf$ in a tight $M$-framelet filter bank must satisfy
\begin{equation} \label{tffb:lpf}
\sum_{\xi\in \Omega_{M}} |\wh{\lpf}(\gbo+2\pi \xi)|^2\le 1, \qquad \forall\, \gbo\in \R^d.
\end{equation}
If the above inequality becomes an identity for all $\gbo\in \R^d$, then the low-pass filter $\lpf$ is called \emph{an orthonormal $M$-wavelet filter}. If $\{\lpf; \hpf_{1},\ldots, \hpf_{L}\}$ is an orthonormal $M$-wavelet filter bank, then $\lpf$ must be an orthonormal $M$-wavelet filter and its corresponding $\{\phi; \psi_1,\ldots, \psi_L\}$ in \eqref{phi:psi}
is a tight $M$-framelet for $L_2(\R^d)$ but it may fail to be an orthonormal $M$-wavelet for $L_2(\R^d)$ (\cite{Dau}). For a filter bank $\{\lpf; \hpf_{1},\ldots, \hpf_{L}\}$ with $L=|\det(M)|-1$ and $\wh{\lpf}(0)=1$,
$\{\phi; \psi_1,\ldots, \psi_L\}$ in \eqref{phi:psi} is an orthonormal $M$-wavelet for $L_2(\R^d)$ if and only if $\{\lpf; \hpf_{1},\ldots, \hpf_{L}\}$ is an orthonormal $M$-wavelet filter bank and $\sm(a,M)>0$, where the technical quantity $\sm(a,M)$ is defined in \eqref{sm:a}. See \cite{BW,CD_1993,Dau,Han:simaa:2003,Han:jat:2003,Han:acha:2004,Han:jcam:2011,Rie_S1,Rie_S2} and references therein for orthonormal wavelets.
For a $d\times d$ dilation matrix $M$, it is trivial to see that $|\det(M)|\ge 2$.
For $|\det(M)|=2$, an orthonormal $M$-wavelet filter bank $\{\lpf; \hpf_{1},\ldots, \hpf_{L}\}$ with $L=|\det(M)|-1$ has only one high-pass filter $\hpf_1$ which is derived from the low-pass filter $\lpf$ by
\begin{equation} \label{canonical}
\wh{\hpf_1}(\gbo):=e^{-i\gbo\cdot \gamma} \ol{\wh{\lpf}(\gbo+2\pi \xi)},\qquad \gbo\in \R^d \quad \mbox{with}\quad \gamma\in \Z^d\backslash [M\Z^d],\; \xi\in \Omega_{M}\backslash\{0\}.
\end{equation}
Therefore, for a dilation matrix $M$ with $|\det(M)|=2$,
an orthonormal $M$-wavelet $\{\phi; \psi_{1},\ldots, \psi_{L}\}$ with $L=|\det(M)|-1$ has only one wavelet function $\psi_{1}$.
Hence, it is of interest in both theory and application to consider dilation matrices $M$ with $|\det(M)|=2$. This is another motivation for us to consider the quincunx dilation matrices in \eqref{quincunxmatrix}.

%Though most of our results in this paper can be stated for any dimension, for simplicity, we restrict ourselves to dimension two in this paper and consider only the quincunx dilation matrix $M_{\sqrt{2}}$.

Due to the importance of high dimensional problems, multivariate wavelets and framelets have been studied for many years now. For example, quincunx orthonormal wavelets have been investigated in \cite{CD_1993,Han:laa:2002} and quincunx biorthogonal wavelets have been studied in \cite{CD_1993,HanJia:MCOM:2000, Jiang_bio_quad_2011}. Using the dilation matrix $M_{\sqrt{2}}$ and perturbation of the Daubechies orthonormal wavelets, a family of quincunx orthonormal wavelets with arbitrarily smoothness orders has been reported in \cite{BW}.
However, compactly supported continuous quincunx orthonormal wavelets cannot have symmetry (see \cite{CD_1993} and \cite[Proposition~2.2]{Han:acha:2004}). Moreover, it still remains unknown so far whether there exists a $C^1$ compactly supported orthonormal $N_{\sqrt{2}}$-refinable function (\cite{CD_1993} and \cite[Example~3.6]{Han:laa:2002}).
In fact, if the dilation matrix $M_{\sqrt{2}}$ is changed into $N_{\sqrt{2}}$ for the family of quincunx wavelet filter banks in \cite{BW}, as a known phenomenon observed in \cite{CD_1993}, their smoothness orders are no more than one and however decreases to zero. The quincunx biorthogonal wavelets constructed in some literature such as \cite{HanJia:MCOM:2000,Jiang_bio_quad_2011} have nice smoothness and/or full $D_4$-symmetry. However the biorthogonal wavelets usually have large supports and the corresponding wavelet transforms have large condition numbers. Pairs of quincunx dual frames have been obtained in \cite[Corollary~3.4]{EH} having only three wavelet functions without symmetry.
Due to the difficulty in constructing multivariate wavelets with desirable properties such as symmetry, short support and high vanishing moments (see \cite{CD_1993,Dau,Han:laa:2002,Han:acha:2004,Han:jcam:2011} and references therein), the current interest has been focusing on the construction of tight $M$-framelets with various dilation matrices and properties.
Tight $M$-framelets have been studied and constructed in many articles. For example, the topic of wavelet frames has been investigated in \cite{CHS_2002,Dau,Daubechies2003,Han_1997,Han:acha:2012,Rie_S2, RS1} and references therein.
The theory and construction of one-dimensional tight $2$-framelets are quite complete so far, for example, see
\cite{Chui_He_2000,CHS_2002,Daubechies2003,Dong_Shen_2007,Dong2010IASNotes,Han:acha:2013,Han:acha:2014,HanMo:simaa:2004,Jiang_2003, MoZhuang:LAA, RS1} and many references therein.
In particular, if $\lpf$ is $\{1,-1\}$-symmetric, the construction of symmetric $2$-framelet filter bank $\{\lpf;\hpf_1,\ldots, \hpf_L\}$ with $L=2$ or $L=3$ and with short support have been completely solved in
\cite{Han:acha:2013,Han:acha:2014,HanMo:simaa:2004} with efficient algorithms.
The construction of multivariate tight framelets
has been reported in \cite{GRon:pams:1998,Han_1997,Han:laa:2002,Han:jcam:2003,Han:mmnp:2014,Jiang_2009,JS_2014,Lai_S_2006,RS_1998} and references therein. The applications of tight framelets to various applications such as image restoration have been investigated in
\cite{Dong2010IASNotes,HanZhao:siims:2014,ShenICM2010,TSZ_2013}. Recently, wavelet frames have been used for surface processing \cite{DJLS_2014,Shen_L_04}. Furthermore, the connections of  wavelet frame based, especially spline tight wavelet  frames based, approach for image restoration to PDE based methods have been established in \cite{CDOS_2012} for the total variational method and extended in \cite{DJS_2013} for the nonlinear diffusion partial differential equation based methods, as well as  in \cite{CDS_2014} for variational
models on the space of piecewise smooth functions.

We now explain our motivations to study quincunx tight framelets and quincunx tight framelet filter banks. From the viewpoint of theory and application for particular areas such as computer aided geometric design and image processing, the following are some key desirable features of a tight $M$-framelet filter bank $\{\lpf; \hpf_1,\ldots,\hpf_L\}$:

\begin{enumerate}
\item[(i)] The high-pass filters $\hpf_{1},\ldots,\hpf_{L}$ have desired high orders of vanishing moments.

\item[(ii)] The low-pass filter $\lpf$ has full symmetry and all the high-pass filters $\hpf_{1},\ldots,\hpf_{L}$ possess desired symmetry.

\item[(iii)] The number $L$ of high-pass filters should be relatively small for computational efficiency.

\item[(iv)] The low-pass filter $\lpf$ should have short support, while the supports of all high-pass filters $\hpf_{1},\ldots,\hpf_{L}$ should not be larger than the support of the low-pass filter $\lpf$.

\item[(v)] The smoothness exponent $\sm(a,M)$ (see \eqref{sm:a}) can be arbitrarily large.
\end{enumerate}

Let $\{\phi;\psi_1,\ldots,\psi_L\}$ be its associated tight $M$-framelet for $L_2(\R^d)$, where $\phi,\psi_1,\ldots,\psi_L$ are defined in \eqref{phi:psi}.
Item (i) implies that all the wavelet generators $\psi_1,\ldots,\psi_L$ have high orders of vanishing moments. The high order of vanishing moments in item (i) is closely related to sparse approximation by tight framelets and necessarily requires that the low-pass filter $\lpf$ should have high order of sum rules.
Item (v) implies that the smoothness exponents of all the functions $\phi,\psi_1,\ldots,\psi_L$ can be arbitrarily large since $\sm(\phi)\ge \sm(a,M)$ and $\sm(\psi_1)=\cdots=\sm(\psi_L)=\sm(\phi)$ (see \eqref{sm:phi}). The definitions of vanishing moments $\vmo(a)$, sum rules $\sr(a,M)$, and smoothness exponents $\sm(\phi)$ and $\sm(a,M)$ will be defined in Section~2.
High orders of vanishing moments in item (i) and smoothness in item (v) are of theoretical interest and importance for characterizing function spaces by framelets.
Item (ii) implies that all the functions $\phi,\psi_1,\ldots,\psi_L$ have symmetry.
The symmetry property in item (ii) is indispensable for applications of tight framelets to certain areas such as computer graphics and is often strongly desired in areas such as image processing for better visual quality. Item (iv) implies that all $\phi,\psi_1,\ldots,\psi_L$ have shortest possible support. Items (iii) and (iv) are important in applications for computational efficiency.
We also point out here that because different applications require different desirable properties of framelets and wavelets, it is not surprising that the above outlined desirable properties in items (i)--(v) may not be needed or should be changed accordingly for a particular application. For example, instead of high orders of vanishing moments in item (i), consecutive orders of vanishing moments starting from vanishing moment one are found to be very useful in image processing \cite{CDOS_2012,Dong2010IASNotes,ShenICM2010}. To achieve directionality in \cite{Han:mmnp:2013,HanZhao:siims:2014} for applications of complex tight framelets in image/video denoising, symmetry of the high-pass filters in item (ii) is sacrificed  (but the low-pass filter has symmetry and the high-pass filters have pairwise symmetry).
Nevertheless, the outlined properties in items (i)--(v) are highly desired for applications in computer graphics, computer aided geometric design as well as other applications.

Despite numerous effort by many researchers on constructions of multivariate tight $M$-framelets and tight $M$-framelet filter banks in many papers, none of them can really achieve all the above desirable properties in items (i)--(v).
For example, tight $M$-framelet filter banks with short supports have been constructed in \cite{GRon:pams:1998,RS_1998} from a special class of almost separable low-pass filters. For a $d$-dimensional filter $\lpf\in l_0(\Z^d)$, we say that $\lpf$ is \emph{an almost separable filter} if its symbol is a finite product of symbols of one-dimensional filters as follows:
\begin{equation}\label{separablefilter}
\wh{\lpf}(\gbo)=\prod_{\ell=1}^K \wh{\lpf_\ell}(\gamma_\ell\cdot\gbo),\qquad \gbo\in \R^d \quad \mbox{with}\quad \lpf_\ell\in l_0(\Z), \gamma_\ell \in \Z^d.
\end{equation}
Because the one-dimensional filters $\lpf_\ell$ used in \cite{GRon:pams:1998,RS_1998} are Haar type low-pass filters with sum rule order one, it is not surprising that all the constructed tight framelets in \cite{GRon:pams:1998,RS_1998} have only one vanishing moment.
For every $d\times d$ dilation matrix $M$,
tight $M$-framelet filter banks with arbitrarily high vanishing moments have been reported in \cite{Han:laa:2002,Han:jcam:2003} by employing the simple observation in \eqref{same:dilation} on the role of a dilation matrix $M$ in a tight $M$-framelet filter bank.
Note that every dilation matrix $M$ can be written as $M=E\Lambda F$ (see \cite{Han:laa:2002,Han:jcam:2003}), where $E,\Lambda, F$ are integer matrices such that $|\det(E)|=|\det(F)|=1$ and $\Lambda$ is diagonal. This allows \cite[Theorem~1.1 and Lemma~3.1]{Han:jcam:2003} and \cite[Corollary~3.4]{Han:laa:2002} to trivially have a tight $\Lambda$-framelet (or orthonormal $\Lambda$-wavelet) filter bank $\{\lpf;\hpf_1,\ldots,\hpf_L\}$ with arbitrarily high vanishing moments and short support through tensor product of one-dimensional ones and consequently, $\{\lpf(E\cdot);\hpf_1(E\cdot),\ldots,\hpf_L(E\cdot)\}$ is a
tight $M$-framelet (or orthonormal $M$-wavelet) filter bank. Note that $\lpf(E\cdot)$ is an almost separable filter by $\wh{\lpf(E\cdot)}(\gbo)=\wh{\lpf}((E^\tp)^{-1}\gbo)$.
Tight $M$-framelet filter banks derived from almost separable low-pass filters can be also trivially constructed in \cite{Han:mmnp:2014} through projecting tensor product tight framelet filter banks. In particular, tight $2I_d$-framelet filter banks for every box spline filter having at least order one sum rule can be painlessly constructed (see \cite[Theorem~2.5]{Han:mmnp:2014}). In fact, all the constructions in \cite{GRon:pams:1998,Han:laa:2002,Han:jcam:2003,Han:mmnp:2014,RS_1998} can be regarded as various special cases of the projection method developed in \cite{Han:mmnp:2014}.
Using sum of squares, for a (two-dimensional) low-pass filter $\lpf$ satisfying \eqref{tffb:lpf}, a general method has been proposed in \cite{CPSS:ca:2013,Lai_S_2006}.
%In the paper,  a method  is given so that for an arbitrary  give box spline in any dimension whose refinement mask satisfies sum rule, we can construction a tight frame system such that  each of wavelet has the same support as the refinable function, and the number of wavelet is the same as the  number of none zero coefficients minus 1.
From any box-spline filter $\lpf$ having at least order one sum rules, recently \cite{FJS:mc:2015} constructs a tight $2I_d$-framelet filter bank whose high-pass filters have short support as that of the low-pass filter $\lpf$ and the number $L-1$ is equal to the number of nonzero coefficients in $\lpf$. But all the constructed tight $2I_d$-framelet filter banks in \cite{FJS:mc:2015} cannot have more than one vanishing moment, since the method in \cite{FJS:mc:2015} requires a low-pass filter to have nonnegative coefficients.
However, all the constructed tight framelets in \cite{CPSS:ca:2013,FJS:mc:2015,GRon:pams:1998,Han:laa:2002,Han:jcam:2003,Han:mmnp:2014,Lai_S_2006,RS_1998}
either lack symmetry or have a very large number $L$ of high-pass filters, while
the supports of the constructed high-pass filters in \cite{Lai_S_2006} could be much larger than the support of the low-pass filter.
Beyond the above constructions of multivariate tight $M$-framelet filter banks, particular examples of tight $M$-framelet filter banks have been given in \cite{Jiang_2009,JS_2014} and other references. However, it remains unclear whether one can construct a family of tight $M$-framelet filter banks (in particular, for $M=M_{\sqrt{2}}$ due to \cite[Theorem~2]{Han:stmalo:2003} on all dilation matrices compatible with the symmetry group $D_4$) achieving all the desirable properties in items (i)--(v).

By \eqref{tffb}, the equations for
a tight $M_{\sqrt{2}}$-framelet filter bank $\{\lpf;\hpf_1,\ldots,\hpf_L\}$ become
\begin{align}
&|\wh{\lpf}(\gbo)|^2+|\wh{\hpf_1}(\gbo)|^2+\sum_{\ell=2}^L |\wh{\hpf_\ell}(\gbo)|^2=1,\label{tffb:1}\\
&\ol{\wh{\lpf}(\gbo)}\wh{\lpf}(\gbo+(\pi,\pi))+\ol{\wh{\hpf_1}(\gbo)} \wh{\hpf_1}(\gbo+(\pi,\pi))+\sum_{\ell=2}^L \ol{\wh{\hpf_\ell}(\gbo)}\wh{\hpf_\ell}(\gbo+(\pi,\pi))=0.\label{tffb:0}
\end{align}
If in addition the following relation (which is a special case of \eqref{canonical}) holds:
\begin{equation}\label{canonical:1}
\wh{\hpf_1}(\gbo)=e^{-i\gbo\cdot (1,0)} \ol{\wh{\lpf}(\gbo+(\pi,\pi))},\qquad \gbo\in \R^2,
\end{equation}
we call $\{\lpf;\hpf_1,\ldots,\hpf_L\}$ \emph{a canonical quincunx tight framelet filter bank}.
Moreover, if $\{\lpf;\hpf_1,\ldots,\hpf_{2s-1}\}$ is a tight $M_{\sqrt{2}}$-framelet filter bank satisfying \eqref{canonical:1} and
\begin{equation} \label{canonical:2}
\wh{\hpf_{2\ell+1}}(\gbo)=e^{-i\gbo\cdot (1,0)} \ol{\wh{\hpf_{2\ell}}(\gbo+(\pi,\pi))},\qquad
\ell=1,\ldots,s-1,\; \gbo\in \R^2,
\end{equation}
then it is called \emph{an $s$-multiple canonical quincunx tight framelet filter bank}.
In particular, for $s=2$, it is called \emph{a double canonical quincunx tight framelet filter bank}. Note that the particular vector $(1,0)$ in \eqref{canonical:1} and \eqref{canonical:2} can be replaced by any vector from $\Z^2\backslash [M_{\sqrt{2}}\Z^2]$. Also note that \eqref{canonical:1} is equivalent to
\[
\hpf_1(k_1,k_2)=(-1)^{1+k_1+k_2} \ol{\lpf(1-k_1,-k_2)},\qquad k_1,k_2\in \Z
\]
and \eqref{canonical:2} is equivalent to
\[
\hpf_{2\ell+1}(k_1,k_2)=(-1)^{1+k_1+k_2} \ol{\hpf_{2\ell}(1-k_1,-k_2)},\qquad k_1,k_2\in \Z,\, \ell=1,\ldots,s-1.
\]

The goal of this paper is to construct a family of quincunx tight framelet filter banks
achieving all the above desirable properties in items (i)--(v) with the additional canonical property in
\eqref{canonical:1} and \eqref{canonical:2}.
For an $s$-multiple canonical quincunx tight framelet filter bank $\{\lpf;\hpf_1,\ldots,\hpf_{2s-1}\}$, the conditions in \eqref{canonical:1} and \eqref{canonical:2} automatically imply \eqref{tffb:0} with $L=2s-1$. Hence, $\{\lpf;\hpf_1,\ldots,\hpf_{2s-1}\}$ is
an $s$-multiple canonical quincunx tight framelet filter bank if and only if
\begin{equation}\label{sumofsquare}
\sum_{\ell=1}^{s-1} \Big[|\wh{\hpf_{2\ell}}(\gbo)|^2+|\wh{\hpf_{2\ell}}(\gbo+(\pi,\pi))|^2\Big]=1-|\wh\lpf(\gbo)|^2-|\wh\lpf(\gbo+(\pi,\pi))|^2,
\end{equation}
which is simply a problem of sum of squares.
If $\{\lpf;\hpf_1,\ldots,\hpf_L\}$ is a canonical quincunx tight framelet filter bank satisfying \eqref{canonical:1} and if $\lpf$ is not an orthonormal $M_{\sqrt{2}}$-wavelet filter, then it is quite trivial to show that $L\ge 3$. Indeed, if $L=1$, then $\{\lpf; \hpf_1\}$ must be an orthonormal $M_{\sqrt{2}}$-wavelet filter bank and consequently, $\lpf$ must be an orthonormal $M_{\sqrt{2}}$-wavelet filter, which is a contradiction to our assumption on $\lpf$. Hence, $L\ge 2$. Suppose that $L=2$. By \eqref{canonical:1}, the equation in \eqref{tffb:0} with $L=2$ becomes $\ol{\wh{\hpf_2}(\gbo)}\wh{\hpf_2}(\gbo+(\pi,\pi))=0$, from which we must have $\hpf_2=0$. This implies $L=1$, a contradiction. Therefore, we must have $L\ge 3$.
On the other hand, as shown in \cite{Han:acha:2013,HanMo:simaa:2004}, there is a very restrictive necessary and sufficient condition for a symmetric tight $2$-framelet filter bank $\{\lpf; \hpf_1,\ldots,\hpf_L\}$ with $L=2$. Due to similar reasons, it is natural that $L=3$ is the smallest possible number of high-pass filters for a symmetric quincunx tight framelet filter bank $\{\lpf; \hpf_1,\ldots,\hpf_L\}$.
One of the main goals of this paper is to construct a family of double canonical quincunx tight framelet filter banks $\{\lpf;\hpf_1,\hpf_2,\hpf_3\}$ with symmetry, short supports, and increasing orders of vanishing moments achieving all the desirable properties in items (i)--(v).

%For a sequence $u\in l_0(\Z^2)$, we can split $u$ into two parts: the even part $u^{[e]}$ of $u$ on the quincunx lattice $M_{\sqrt{2}}\Z^2$, and the odd part $u^{[o]}$ of $u$ on $\Z^2\backslash [M_{\sqrt{2}}\Z^2]$ as follows:
%
%\begin{equation}\label{cosetfilter}
%u^{[e]}(k):=\frac{1}{\sqrt{2}}u(M_{\sqrt{2}} k) \quad \mbox{and}\quad u^{[o]}(k):=\frac{1}{\sqrt{2}} u(M_{\sqrt{2}}k+(1,0)),\qquad k\in \Z^2.
%\end{equation}
%
%Due to the relation $\wh{u}(\gbo)=\frac{1}{\sqrt{2}}(\wh{u^{[e]}}(M_{\sqrt{2}}\gbo)+\wh{u^{[o]}}(M_{\sqrt{2}}\gbo)e^{-i \gbo\cdot(1,0)})$, one can directly check that
%$|\wh{u^{[e]}}(\gbo)|^2+|\wh{u^{[o]}}(\gbo)|^2=|\wh{u}(M_{\sqrt{2}}^\tp\gbo)|^2+
%|\wh{u}(M_{\sqrt{2}}^\tp\gbo+(\pi,\pi))|^2$.
%Consequently, the condition in \eqref{sumofsquare} for
%an $s$-multiple canonical tight $M_{\sqrt{2}}$-framelet filter bank $\{\lpf;\hpf_1,\ldots,\hpf_{2s-1}\}$ can be rewritten as
%
%\begin{equation}\label{sumofsquare:2}
%\sum_{\ell=1}^s|\wh{\hpf_{2\ell}^{[e]}}(\gbo)|^2+|\wh{\hpf_{2\ell}^{[o]}}(\gbo)|^2
%=1-|\wh{\lpf}(M_{\sqrt{2}}^{-\tp}\gbo)|^2-
%|\wh{\lpf}(M_{\sqrt{2}}^{-\tp}\gbo+(\pi,\pi))|^2.
%\end{equation}
%

%\textbf{Add more motivations below for using the special structure in \eqref{canonical:1} and connections to Riesz wavelets}

The structure of the paper is as follows. In Section~\ref{sec:double-canonical}, we shall first introduce a family of symmetric minimally supported two-dimensional low-pass filters with arbitrarily high sum rule orders and linear-phase moments.
% Such a family of low-pass filters is of particular interest in their applications to computer graphics and computer aided geometric design, due to their polynomial preservation property, short support and high smoothness.
Then we shall employ such symmetric low-pass filters to construct a family of compactly supported tight framelets with double canonical quincunx tight framelet filter banks $\{\lpf;\hpf_1,\hpf_2,\hpf_3\}$ with symmetry and arbitrarily high orders of vanishing moments. Numerical calculation also indicates that the smoothness exponents of this family of compactly supported tight framelets can be arbitrarily large.
In Section~\ref{sec:double-canonical:2}, we shall generalize the particular construction in Section~\ref{sec:double-canonical} and propose a general construction of double canonical quincunx tight framelet filter banks with symmetry and vanishing moments which are derived from one-dimensional filters with linear-phase moments.
A few illustrative examples of such double canonical quincunx tight framelet filter banks $\{\lpf;\hpf_1,\hpf_2,\hpf_3\}$ are given in Sections~\ref{sec:double-canonical} and~\ref{sec:double-canonical:2}.
In Section~\ref{sec:multiple-canonical}, we shall take another approach by studying multiple  canonical symmetric quincunx tight framelet filter banks using almost separable low-pass filters. In particular, we present a family of compactly supported $6$-multiple canonical real-valued quincunx tight framelets and a family of compactly supported double canonical complex-valued quincunx tight framelets such that both of them have symmetry and arbitrarily high orders of smoothness exponents and vanishing moments.
We complete this paper by providing a detailed proof to Theorems~\ref{thm:a2D2n2n} and~\ref{thm:sm} in Appendix~\ref{sec:proofs}.

\section{Double canonical symmetric quincunx tight framelets with minimal support}
\label{sec:double-canonical}

In this section we shall first discuss how to construct a family of minimally supported symmetric low-pass filters with increasing orders of sum rules and linear-phase moments.
Such a family of low-pass filters is of particular interest in their applications to computer graphics and computer aided geometric design, due to their polynomial preservation property, short support and high smoothness.
Then we shall use such low-pass filters to build double canonical symmetric quincunx tight framelet filter banks with increasing order of vanishing moments.

For an integer $j$ such that $1\le j\le d$, by $\partial_j$ we denote the partial derivative with respect to the $j$th coordinate of $\R^d$. Define $\N_{0}:=\N \cup\{0\}$. For any $\mu=(\mu_1, \ldots, \mu_d)\in \N_{0}^d$, we define $|\mu|:=|\mu_1|+\cdots+|\mu_d|$ and $\partial^\mu$ the differentiation operator $\partial_1^{\mu_1}\cdots \partial_d^{\mu_d}$. For a nonnegative integer $m$ and two smooth functions $f, g$, we shall use the following big $\bo$ notation
\begin{equation} \label{deriv:O}
f(\gbo)=g(\gbo)+\bo(\|\gbo-\gbo_0\|^m), \qquad \gbo\to \gbo_0
\end{equation}
to mean the following relation:
\begin{equation} \label{deriv:1}
\partial^\mu f(\gbo_0)=\partial^\mu g(\gbo_0), \qquad \forall\; \mu\in \N_0^d \;\; \mbox{satisfying}\;\; |\mu|<m.
\end{equation}
For smooth functions, as shown in \cite[Lemma~1]{Han:jcam:2011},
using the big $\bo$ notation in \eqref{deriv:O} to mean \eqref{deriv:1} agrees with the commonly accepted big $\bo$ notation in the literature.

Let $\lpf\in l_0(\Z^d)$ be a filter.
We say that the filter $a$ has \emph{order $m$ sum rules} with respect to a dilation matrix $M$ if $\wh a(0)=1$ and $\wh{\lpf}(\gbo+2\pi \xi)=\bo(\|\gbo\|^m)$ as $\gbo\to 0$ for all $\xi\in \Omega_{M}\backslash \{0\}$. In particular, we define $\sr(a,M):=m$ with $m$ being the largest such integer.
We say that the filter $a$ has \emph{order $n$ vanishing moments} if $\wh{\lpf}(\gbo)=\bo(\|\gbo\|^n)$ as $\gbo\to 0$. In particular, we define $\vmo(a):=n$ with $n$ being the largest such integer. We say that a filter $\lpf\in l_0(\Z^d)$ has \emph{order $n$ linear-phase moments with phase $\mathbf{c}\in \R^d$} if $\wh{\lpf}(\gbo)=e^{-i\mathbf{c}\cdot \gbo}+\bo(\|\gbo\|^n)$ as $\gbo\to 0$. In particular, we define $\lpm(a):=n$ with $n$ being the largest such integer.
% and $\mathbf{c}:=\sum_{k\in \Z^d} \lpf(k) k$.
The notion of linear-phase moments has been introduced in \cite{Han:adv:2010} for studying symmetric complex orthonormal $2$-wavelets and plays a central role in the construction of complex symmetric orthonormal wavelets, subdivision schemes with polynomial preservation property in computer graphics, and symmetric tight framelets with vanishing moments (see \cite{Dong_Shen_2007,DHSS:jat:2008,Han:adv:2010,Han:jcam:2011,Han:mmnp:2013}).

Suppose that $\{\lpf;\hpf_1,\ldots,\hpf_L\}$ is a tight $M$-framelet filter bank. Through the equations in \eqref{tffb} and assume that $\wh{\lpf}(0)=1$, it is shown (see e.g. \cite{Daubechies2003,Han:mmnp:2013}) that
\begin{equation}\label{tffb:vm}
\min(\vmo(\hpf_1),\ldots,\vmo(\hpf_L))=\min(\sr(\lpf,M),\tfrac{1}{2}\lpm(\lpf*\lpf^\star)),
\end{equation}
where $\wh{\lpf*\lpf^\star}(\gbo):=|\wh{\lpf}(\gbo)|^2$. It is straightforward to see that $\lpm(\lpf*\lpf^\star)\ge \lpm(\lpf)$.
If the low-pass filter $\lpf$ is symmetric about a point $\mathbf{c}\in \R^d$:
$\lpf(2\mathbf{c}-k)=\ol{\lpf(k)}$ for all $k\in \Z^d$,
it has been shown in \cite[Proposition~5.3]{Han:mmnp:2013} that $\lpm(\lpf*\lpf^\star)=\lpm(\lpf)$ and for a tight $M$-framelet filter bank $\{\lpf;\hpf_1,\ldots,\hpf_L\}$ with $\wh{\lpf}(0)=1$,
\begin{equation}\label{lpf:lpm}
\min(\vmo(\hpf_1),\ldots,\vmo(\hpf_L))=\min(\sr(a,M),\tfrac{1}{2}\lpm(\lpf)).
\end{equation}
Therefore, to construct quincunx tight framelet filter banks with symmetry and high vanishing moments, it is necessary to have low-pass filters having high orders of sum rules and linear-phase moments.

The following result presents a family of minimally supported $D_4$-symmetric low-pass filters having increasing orders of sum rules and linear-phase moments.

\begin{theo}\label{thm:a2D2n2n}
For every positive integer $n$, there exists a unique two-dimensional filter $\lpf_{2n,2n}^{2D}$ such that $\lpf_{2n,2n}^{2D}$ is supported inside $[1-n,n]^2\cap\Z^2$, has order $2n$ sum rules with respect to $M_{\sqrt{2}}$ and order $2n$ linear-phase moments with phase $\mathbf{c}:=(1/2,1/2)$. Moreover,
\begin{enumerate}
\item[{\rm(i)}] the filter $\lpf^{2D}_{2n,2n}$ is real-valued and is given by
\begin{equation} \label{a2D2n2n}
\wh{\lpf^{2D}_{2n,2n}}(\omega_1,\omega_2)=\frac{1}{2}
[\wh{u}(\omega_1+\omega_2)+\wh{u}(\omega_1-\omega_2)e^{-i\omega_2}],
\end{equation}
where $\wh{u}(\omega):=(\wh{a^{I}_{2n}}(\omega/2)-\wh{a^I_{2n}}(\omega/2+\pi))e^{-i\omega/2}$ and $a^{I}_{2n}$ is the interpolatory $2$-wavelet filter given by
\begin{equation}\label{aI2n}
\wh{a^I_{2n}}(\omega):=\cos^{2n}(\omega/2) \sum_{j=0}^{n-1} \binom{n-1+j}{j} \sin^{2j}(\omega/2),\quad \omega\in\R;
\end{equation}
\item[{\rm(ii)}] the filter $\lpf^{2D}_{2n,2n}$ is $D_4$-symmetric about the point $\mathbf{c}=(1/2,1/2)$;
\item[{\rm(iii)}] $\phi^{N_{\sqrt{2}}}=\phi^{M_{\sqrt{2}}}(\cdot+(1,1))$ and
$\phi^{M_{\sqrt{2}}}$ is real-valued with the following symmetry property:
\begin{equation}\label{phi:sym}
\phi^{M_{\sqrt{2}}}(E(\cdot-\mathbf{c}_\phi)+\mathbf{c}_\phi)=\phi^{M_{\sqrt{2}}},\qquad \forall\, E\in D_4
\end{equation}
%
% \quad \mbox{with}\quad
with $\mathbf{c}_\phi:=(M_{\sqrt{2}}-I_2)^{-1}\mathbf{c}=(3/2,1/2)$,
where $\phi^{M_{\sqrt{2}}}$ and $\phi^{N_{\sqrt{2}}}$ are the refinable functions associated with the filter $\lpf$ and the dilation matrices $M_{\sqrt{2}},N_{\sqrt{2}}$ in \eqref{quincunxmatrix}, respectively, and are defined in the frequency domain through
\begin{equation} \label{phi:psi:a2Da2n2n}
\wh{\phi^{M_{\sqrt{2}}}}(\gbo):=\prod_{j=1}^\infty \wh{\lpf^{2D}_{2n,2n}}((M_{\sqrt{2}}^\tp)^{-j}\gbo)\quad\mbox{and}\quad \wh{\phi^{N_{\sqrt{2}}}}(\gbo):=\prod_{j=1}^\infty \wh{\lpf^{2D}_{2n,2n}}((N_{\sqrt{2}}^\tp)^{-j}\gbo),\qquad \gbo\in \R^2.
\end{equation}
\end{enumerate}
\end{theo}

The proof of Theorem~\ref{thm:a2D2n2n} is given in Appendix~\ref{sec:proofs}.
We now derive double canonical quincunx symmetric tight framelet filter banks from the low-pass filters $\lpf^{2D}_{2n,2n}$ constructed in Theorem~\ref{thm:a2D2n2n}.

\begin{theo}\label{thm:b2D2n2n}
Let $\lpf=\lpf^{2D}_{2n,2n}$ with $n\in \N$ be the filter constructed in \eqref{a2D2n2n} of Theorem~\ref{thm:a2D2n2n}. Define a high-pass filter $\hpf_2$ by
\begin{equation}\label{def:b2}
\wh{\hpf_2}(\omega_1,\omega_2):=\frac{1}{2}[\wh{v}(\go_1+\go_2)+\wh{v}(\go_1-\go_2)e^{-i\go_2}]
\quad \mbox{with}\quad \wh{v}(\go):=2\wh{\lpf^{D}_n} (\omega/2)\wh{\lpf^{D}_n} (\omega/2+\pi),
\end{equation}
and define high-pass filters $\hpf_1, \hpf_3$ as in \eqref{canonical:1} and \eqref{canonical:2}, where
$\lpf_n^{D}\in l_0(\Z)$ is a real-valued Daubechies orthonormal $2$-wavelet filter satisfying $|\wh{\lpf^{D}_{n}}(\omega)|^2=\wh{\lpf^I_{2n}}(\omega)$.
Then $\{\lpf; \hpf_1,\hpf_2,\hpf_3\}$ is a double canonical quincunx tight framelet filter bank satisfying
\begin{enumerate}
\item[{\rm(i)}] all high-pass filters $\hpf_1,\hpf_2,\hpf_3$ have real coefficients and the following symmetry:
\begin{equation}\label{sym:b1}
\hpf_1(E(k-\mathring{\mathbf{c}})+\mathring{\mathbf{c}})=\det(E) \hpf_1(k),\qquad \forall\, k\in \Z^2, E\in D_4 \quad \mbox{with}\quad \mathring{\mathbf{c}}:=(1/2,-1/2)
\end{equation}
and
\begin{equation}\label{sym:b2b3}
\hpf_2(k_1,1-k_2)=\hpf_2(k_1,k_2)\quad \mbox{and}\quad
\hpf_3(k_1,-1-k_2)=-\hpf_3(k_1,k_2),\qquad \forall\, k_1, k_2\in \Z;
\end{equation}

\item[{\rm(ii)}]  all high-pass filters $\hpf_1,\hpf_2,\hpf_3$  have at least order $n$ vanishing moments;
\item[{\rm(iii)}]  the supports of $\hpf_1, \hpf_2, \hpf_3$ are no larger than that of the low-pass filter $\lpf$.
\end{enumerate}
 Moreover, $\{\phi^{M_{\sqrt{2}}}; \psi_1,\psi_2,\psi_3\}$ is a tight $M_{\sqrt{2}}$-framelet in $L_2(\R^2)$ such that $\phi^{M_{\sqrt{2}}}$ has the symmetry in \eqref{phi:sym},
\begin{equation}\label{psi1:sym}
\psi_1(E(\cdot-\mathbf{c}_1)+\mathbf{c}_1)=\det(E) \psi_1,\qquad \forall\, E\in D_4\quad \mbox{with}\quad \mathbf{c}_1:=(1,1)
\end{equation}
and
\begin{equation}\label{psi2psi3:sym}
\psi_2(x_2+1,x_1-1)=\psi_2(x_1,x_2),\qquad
\psi_3(x_2,x_1)=-\psi_3(x_1,x_2),
\end{equation}
where $\phi^{M_{\sqrt{2}}}$ is defined in \eqref{phi:psi:a2Da2n2n} and $\wh{\psi_\ell}(\gbo):=\wh{\hpf_\ell}((M_{\sqrt{2}}^\tp)^{-1}\gbo) \wh{\phi^{M_{\sqrt{2}}}}((M_{\sqrt{2}}^\tp)^{-1}\gbo)$ for $\ell=1,2,3$.
\end{theo}

%\textbf{Check the above symmetry patterns to see whether they are correct or not!}

\begin{proof}
Let $\wh {u}(\go) =(\wh{a^{I}_{2n}}(\omega/2)-\wh{a^I_{2n}}(\omega/2+\pi))e^{-i\omega/2}=
(2\wh{a_{2n}^{I}}(\go/2)-1)e^{-i\go/2}$, where we use $\wh{\lpf^I_{2n}}(\go/2)+\wh{\lpf^I_{2n}}(\go/2+\pi)=1$.
By the definition of $\lpf=\lpf_{2n,2n}^{2D}$ in \eqref{a2D2n2n}, we have
\begin{equation}\label{a:square}
|\wh{\lpf}(\go_1,\go_2)|^2+|\wh{\lpf}(\go_1+\pi,
\go_2+\pi)|^2=\frac{1}{2}\Big[|\wh{u}(\go_1+\go_2)|^2+|\wh{u}(\go_1-\go_2)|^2\Big].
\end{equation}
Similarly, by the definition of $\hpf_2$, we have
\begin{equation}\label{b:square}
|\wh{\hpf_2}(\go_1,\go_2)|^2+|\wh{\hpf_2}(\go_1+\pi,
\go_2+\pi)|^2=\frac{1}{2}\Big[|\wh{v}(\go_1+\go_2)|^2+|\wh{v}(\go_1-\go_2)|^2\Big].
\end{equation}
Since $|\wh{\lpf^D_n}(\go)|^2=\wh{\lpf^I_{2n}}(\go)$, we have
\[
\begin{aligned}
1-|\wh{u}(\go)|^2&=1-|2\wh{\lpf^I_{2n}}(\go/2)-1|^2
= 1-4(\wh{\lpf^I_{2n}}(\go/2))^2+4\wh{\lpf^I_{2n}}(\go/2)-1
\\&= 4\wh{\lpf^I_{2n}}(\go/2)(1-\wh{\lpf^I_{2n}}(\go/2))
=4\wh{\lpf^I_{2n}}(\go/2)\wh{\lpf^I_{2n}}(\go/2+\pi)
=|\wh{v}(\go)|^2.
\end{aligned}
\]
Consequently, \eqref{sumofsquare} holds with $s=2$. Therefore, $\{\lpf;\hpf_1,\hpf_2,\hpf_3\}$ is a  double canonical quincunx tight framelet filter bank.

Since $\lpf$ is $D_4$-symmetric about the point $\mathbf{c}=(1/2,1/2)$, \eqref{def:mask-sym-group-time} is equivalent to
\begin{equation}\label{D4}
\wh{\lpf}(E^\tp \gbo)=e^{i(I_2-E)\mathbf{c}\cdot \gbo}\wh{\lpf}(\gbo), \quad \gbo\in \R^2, E\in D_4.
\end{equation}
%
%Define
%$\wh{\mathring{\lpf}}(\gbo) %:=\wh{\lpf_{2n}^I}(\frac{\go_1+\go_2}{2})+\wh{\lpf_{2n}^I}(\frac{\go_1-\go_2}{2})-1$.
%It is easy to show that $\wh{\lpf}(\gbo) = \wh{\mathring{\lpf}}(\gbo)e^{-i\mathbf{c}\cdot \gbo}$, ${\wh{\mathring{\lpf}}(E^\top \cdot)} =  \wh{\mathring{\lpf}}$,  and ${\wh{\mathring{\lpf}}(\cdot+E^\top (\pi,\pi))} =  \det(E)\wh{\mathring{\lpf}}(\cdot+(\pi,\pi))$ for all $E\in D_4$. Therefore, we can rewrite
%\[
%\wh{\hpf_1}(\gbo) = e^{-i\go_1}\overline{\wh{\lpf}(\gbo+(\pi,\pi))} = e^{-i\go_1}\overline{\wh{\mathring{\lpf}}(\gbo+(\pi,\pi)) e^{-i \mathbf{c}\cdot (\gbo+(\pi,\pi))}}
%= -\overline{\wh{\mathring{\lpf}}(\gbo+(\pi,\pi))} e^{-i \mathring{\mathbf{c}}\cdot \gbo},
%\]
%where $\mathring{\mathbf{c}}= (1/2,-1/2)$. Hence, we have,
%\[
%\begin{aligned}
%\wh{\hpf_1}(E^\top \gbo)
%&  = -\overline{\wh{\mathring{\lpf}}(E^\top\gbo+(\pi,\pi))} e^{-i \mathring{\mathbf{c}}\cdot E^\top\gbo}
% = -\overline{\wh{\mathring{\lpf}}(\gbo+E^\top(\pi,\pi))} e^{-i \mathring{\mathbf{c}}\cdot \gbo}  e^{i \mathring{\mathbf{c}}\cdot (I_2-E^\top)\gbo} \\
%&  = -\det(E)\overline{\wh{\mathring{\lpf}}(\gbo+(\pi,\pi))} e^{-i \mathring{\mathbf{ c}}\cdot \gbo} %e^{i (I_2-E)\mathring{\mathbf{c}}\cdot \gbo}
%= \det(E)\wh{\hpf_1}(\gbo) e^{i (I_2-E)\mathring{\mathbf{c}}\cdot \gbo}
%\end{aligned}
%\]
%for all $E\in D_4$.
For $E\in D_4$, we have $(I-E)(1,1)^\tp\in 2\Z^2$ and by the definition of $b_1$,
\begin{align*}
\wh{\hpf_1}(E^\tp \gbo)&=e^{-i\gbo\cdot E(1,0)} \ol{\wh{\lpf}(E^\tp\gbo+(\pi,\pi))}=
e^{-i\gbo\cdot E(1,0)}\ol{\wh{\lpf}(E^\tp(\gbo+(\pi,\pi)))}\\
&=
e^{-i\gbo\cdot E(1,0)} e^{-i(I-E)\mathbf{c}\cdot(\gbo+(\pi,\pi))}\ol{\wh{\lpf}(\gbo+(\pi,\pi))}=
\det(E) e^{i(I-E)\mathring{\mathbf{c}}\cdot \gbo}\wh{\hpf_1}(\gbo).
\end{align*}
This proves \eqref{sym:b1}.
By the definitions of $\hpf_2$ and $\wh{\hpf_3}(\gbo) = e^{-i\go_1}\ol{\wh{\hpf_2}(\gbo+(\pi,\pi))}$, we have
\[
\wh{\hpf_2}(\go_1,-\go_2) = \wh{\hpf_2}(\go_1,\go_2)e^{i\go_2}\quad \mbox{and}\quad
\hpf_3(\go_1,-\go_2)=-\wh{\hpf_3}(\go_1,\go_2)e^{-i\go_2},
\]
which are equivalent to \eqref{sym:b2b3}. Therefore, item (i) holds.

Item (ii) follows directly from $\min(\vmo(\hpf_1),\vmo(\hpf_2),\vmo(\hpf_3))=\min(\sr(\lpf,M_{\sqrt{2}}),\frac12\lpm(\lpf))=n$ due to $\sr(\lpf,M_{\sqrt{2}})=\lpm(\lpf)=2n$. Item (iii) can be directly checked.

By \cite[Proposition~2.1]{Han:acha:2004}, the identity in \eqref{psi1:sym} follows directly from \eqref{phi:sym} and \eqref{sym:b1}, while the identities in \eqref{psi2psi3:sym} follows directly from \eqref{phi:sym} and \eqref{sym:b2b3}.
\end{proof}

For a function $\phi\in L_2(\R^d)$, its Sobolev \emph{smoothness exponent} $\sm(\phi)$ is defined to be
\begin{equation}\label{sm:phi}
\sm(\phi):=\sup\Big\{ \tau \in \R \; : \; \int_{\R^d} |\wh{\phi}(\xi)|^2(1+\|\xi\|^2)^\tau d\xi<\infty\Big\}.
\end{equation}
%
\begin{comment}
Recall that the Sobolev $L_p$ smoothness of a function $\phi\in L_p(\R^d)$ with $1\le p\le \infty$ is measured by its $L_p$ critical smoothness exponent $\sm(\phi)$:
%
\begin{equation}\label{sm:phi}
\sm_p(\phi):=\sup\Big\{ n+\tau \; : \; \|\partial^\mu\phi-\partial^\mu \phi(\cdot-t)\|_{L_p(\R^d)}\le C_f \|t\|^\tau, \quad \mu\in \N_0^d, |\mu|=n, t\in \R^d\Big\},
\end{equation}
%
where $0\le \tau\le 1$ and $n$ is the largest nonnegative integer such that $\partial^\mu \phi\in L_p(\R^d)$ for all $\mu\in \N_0^d$ with $|\mu|\le n$. In particular, we define $\sm(\phi):=\sm_2(\phi)$ for $p=2$.
\end{comment}
If $\phi$ is an $M$-refinable function associated with a filter $a\in l_0(\Z^d)$, then the smoothness exponent $\sm(\phi)$ is closely linked to a quantity $\sm(a,M)$ introduced in \cite{Han:jat:2003}.
%Recall that the subdivision operator $\mathcal{S}_{a,M}$ is defined to be
%%
%\begin{equation}\label{sd:op}
%[\mathcal{S}_{a,M} v](n):=|\det(M)|\sum_{k\in \Z^d} v(k) a(n-Mk),\qquad n\in \Z^d.
%\end{equation}
%%
For $u\in l_0(\Z^d)$ and $\mu=(\mu_1,\ldots,\mu_d) \in \N_0^d$, we define
\begin{equation} \label{d:diff}
\nabla_k u:=u-u(\cdot-k), \qquad k \in \Z^d \quad \mbox{and}\quad
\nabla^\mu:=\nabla_{e_1}^{\mu_1}\cdots\nabla_{e_d}^{\mu_d},
\end{equation}
where $e_j=(0,\ldots,0,1,0,\ldots,0)\in \R^d$ has its only nonzero entry $1$ at the $j$th coordinate. By $\delta$ we denote \emph{the Dirac sequence} such that $\delta(0)=1$ and $\delta(k)=0$ for all $k\in \Z^d\backslash \{0\}$.
For $a\in l_0(\Z^d)$ and a $d\times d$ dilation matrix $M$, let $m:=\sr(a,M)$. For $1\le p\le \infty$, the \emph{smoothness exponent} $\sm_p(a,M)$ (see \cite{Han:jat:2003}) is defined to be
\begin{equation} \label{sm:a}
\sm_p(a,M):=\tfrac{d}{p}-d\log_{|\det(M)|} \rho_{m}(a,M)_p\qquad \mbox{and}\qquad
\sm(a,M):=\sm_2(a,M),
\end{equation}
where
\begin{equation} \label{rhoa}
\rho_{m}(a,M)_p:=\sup\left\{\lim_{n\to \infty} \|\nabla^\mu \mathcal{S}_{a,M}^n \delta\|_{l_p(\Z^d)}^{1/n} \; : \; \mu\in \N_0^d, |\mu|=m\right\}
\end{equation}
and the subdivision operator $\mathcal{S}_{a,M}$ is defined to be
\begin{equation}\label{sd:op}
[\mathcal{S}_{a,M} v](n):=|\det(M)|\sum_{k\in \Z^d} v(k) a(n-Mk),\qquad n\in \Z^d.
\end{equation}
The quantity $\sm(a,M)$ can be computed by
\cite[Algorithm~2.1]{Han:simaa:2003}.
We say that $M$ is \emph{isotropic} if $M$ is similar to a diagonal matrix $\mbox{diag}(\lambda_1,\ldots,\lambda_d)$ with $|\lambda_1|=\cdots=|\lambda_d|$.
Note that the two quincunx matrices $M_{\sqrt{2}}$ and $N_{\sqrt{2}}$ in \eqref{quincunxmatrix} are isotropic.
For an isotropic dilation matrix $M$, we have $\sm(\phi)\ge \sm(a,M)$ and if in addition the integer shifts of $\phi$ are stable (i.e., $\sum_{k\in \Z^d} |\wh{\phi}(\gbo+2\pi k)|^2\ne 0$ for all $\gbo\in \R^d$), then $\sm(\phi)=\sm(a,M)$ (e.g., see \cite{Han:simaa:2003,Han:jat:2003} and many references therein).

The smoothness exponents $\sm(\lpf^{2D}_{2n,2n},M_{\sqrt{2}})$ and $\sm(\lpf^{I}_{2n},2)$ for $n=1,\ldots,10$ in Table~\ref{tab:smoothness} are calculated by \cite[Algorithm~2.1]{Han:simaa:2003} using $D_4$ symmetry group.
Note that $\sm(\lpf,N_{\sqrt{2}})=\sm(a,M_{\sqrt{2}})$ since $\lpf$ is $D_4$-symmetric.
%The smoothness exponents $\sm(\lpf^{2D}_{2n,2n},M_{\sqrt{2}})$ in Table~\ref{tab:smoothness} indicates that $\sm(\lpf^{2D}_{2n,2n},M_{\sqrt{2}})\approx \sm(\lpf_{2n}^I,2)\to \infty$ as $n\to \infty$.

\begin{table}[ht]
\begin{center}
\begin{tabular}{|c|c|c|c|c|c|c|c|c|c|c|c|c|} \hline
$n$        &$1$    &$2$       &$3$   &$4$  &$5$ &$6$ & $7$ & $8$ & $9$ & $10$\\ \hline
{ $\sm(\lpf^{2D}_{2n,2n},M_{\sqrt{2}})$}  &$2.0$  &$3.0365$  &$3.5457$ &$4.0269$ &$4.4970$ &$4.9658$  & $5.4350$ & $5.9038$ & $6.3714$ & $6.8374$
%&$7.3007$ &$7.7437$
\\ \hline
{ $\sm(\lpf_{2n}^I,2)$} &$1.5$  &$2.4408$  &$3.1751$ &$3.7931$ &$4.3441$ &$4.8620$ &$5.3628$ & $5.8529$ & $6.3352$ &$6.8115$
%&$7.2829$ &$7.7501$
\\ \hline
\end{tabular}
\caption{
The smoothness exponents of the quincunx low-pass filters $\lpf_{2n,2n}^{2D}$ in \eqref{a2D2n2n} and of interpolatory $2$-wavelet filters $\lpf_{2n}^I$ in \eqref{aI2n} for $n=1,\ldots,10$, computed by \cite[Algorithm~2.1]{Han:simaa:2003}. Note that $\sm(\lpf^{2D}_{2n,2n},N_{\sqrt{2}})=\sm(\lpf^{2D}_{2n,2n},M_{\sqrt{2}})$.
}
\label{tab:smoothness}
\end{center}
\end{table}

We complete this section by presenting two examples to illustrate the results in Theorems~\ref{thm:a2D2n2n} and ~\ref{thm:b2D2n2n}.

\begin{example}\label{exam:p_20}\rm{
Take $n=1$ in Theorems~\ref{thm:a2D2n2n} and ~\ref{thm:b2D2n2n}.
Then $\lpf =\lpf_{2,2}^{2D}$ in \eqref{a2D2n2n} with $n=1$ is given by
\[
\wh{\lpf}(\omega_1,\omega_2)=\frac 1{4}(1+e^{-i\go_1})(1+e^{-i\go_2})
\]
%Note that $a$ is the refinement mask for the $4$-directional box spline $B_{1111}$.
%Then $a$ is $D_4$-symmetric with $\sr(a,M_{\sqrt{2}})=2$, $\lpm(a)=2$, and $\sm(a,M_{\sqrt{2}})=2$.
and
\[
\wh{\hpf_1}(\gbo):=e^{-i\go_1}\ol{\wh{\lpf}(\gbo+(\pi,\pi))} = \frac14(1-e^{-i\go_1})(e^{i\go_2}-1).
\]
By
$\wh{\lpf^D_1}(\go):=\frac 12(1+e^{-i\go})$,
we have $\wh{v}(\go):=2\wh{\lpf^D_1}(\go/2)\wh{\lpf^D_1}(\go/2+\pi)=\frac12(1-e^{-i\go})$.
Then
\[
\wh{\hpf_2}(\go_1,\go_2):= \frac{1}{2}(\wh{v}(\go_1+\go_2)+\wh{v}(\go_1-\go_2)e^{-i\go_2})=\frac14(1-e^{-i\go_1})(1+e^{-i\go_2})
\]
and
\[
\wh{\hpf_3}(\gbo):=e^{-i\go_1}\ol{\wh{\hpf_2}(\gbo+(\pi,\pi))} =\frac14(1+e^{-i\go_1})(1-e^{i\go_2}).
\]
The double canonical quincunx tight framelet filter bank $\{\lpf; \hpf_1,\hpf_2,\hpf_3\}$
is given by
\[
a=\frac{1}{4}\left[ \begin{matrix} 1 &1\\ \boxed{1} & 1 \end{matrix}
 \right]_{[0,1]^2},\quad
\hpf_1=\frac{1}{4}\left[ \begin{matrix} \boxed{-1} & 1\\ 1 & -1\end{matrix}
 \right]_{[0,1]\times[-1,0]},\quad
\hpf_2=\frac{1}{4}\left[ \begin{matrix} 1& -1\\ \boxed{1} & -1 \end{matrix} \right]_{[0,1]^2},\quad
\hpf_3=\frac{1}{4}\left[ \begin{matrix} \boxed{1} & 1\\ -1 & -1\end{matrix} \right]_{[0,1]\times[-1,0]}.
\]
Note that $\sr(\lpf,M_{\sqrt{2}})=2$, $\lpm(\lpf)=2$, and $\sm(\lpf,M_{\sqrt{2}})=\sm(\lpf,N_{\sqrt{2}})=2$.
The filter $\lpf$ is $D_4$-symmetric about $(\frac{1}{2},\frac{1}{2})$, while $\hpf_1$ has the symmetry in \eqref{sym:b1} and $\hpf_2, \hpf_3$ have the symmetry in \eqref{sym:b2b3} with $\vmo(\hpf_1)=2$ and $\vmo(\hpf_2)=\vmo(\hpf_3)=1$.
%See Figure~\ref{fig:1} for the graphs of the generators $\phi,\psi_1,\psi_2,\psi_3$ in the tight $M_{\sqrt{2}}$-framelet $\{\phi;\psi_1,\psi_2,\psi_3\}$,
Let $\phi,\psi_1,\psi_2,\psi_3$ be defined in \eqref{phi:psi} with $M=M_{\sqrt{2}}$, $L=3$ and $a=a^{2D}_{2,2}$. Then $\{\phi;\psi_1,\psi_2,\psi_3\}$ is a tight $M_{\sqrt{2}}$-framelet in $L_2(\R^2)$ such that $\phi,\psi_1,\psi_2,\psi_3$ have symmetry property as in  \eqref{phi:sym}, \eqref{psi1:sym}, and \eqref{psi2psi3:sym}.
}\end{example}

\begin{comment}
\begin{figure}
        \centering
        \begin{subfigure}[b]{0.33\textwidth}
                %\includegraphics[width=\textwidth]{Example1/phi_c}
                \includegraphics[width=\textwidth]{Example1/phi}
                \caption{$\phi$}
                \label{fig1:phi}
        \end{subfigure}%
        %
        \begin{subfigure}[b]{0.33\textwidth}
                %\includegraphics[width=\textwidth]{Example1/psi1_c}
                \includegraphics[width=\textwidth]{Example1/psi1}
                \caption{$\psi_1$}
                \label{fig1:psi1}
        \end{subfigure}

        \begin{subfigure}[b]{0.33\textwidth}
                %\includegraphics[width=\textwidth]{Example1/psi2_c}
                \includegraphics[width=\textwidth]{Example1/psi2}
                \caption{$\psi_2$}
                \label{fig1:psi2}
        \end{subfigure}
        %
        \begin{subfigure}[b]{0.33\textwidth}
                %\includegraphics[width=\textwidth]{Example1/psi3_c}
                \includegraphics[width=\textwidth]{Example1/psi3}
                \caption{$\psi_3$}
                \label{fig1:psi3}
        \end{subfigure}
        \caption{Graphs of $\phi,\psi_1,\psi_2,\psi_3$ in Example~\ref{exam:p_20}.}
        \label{fig:1}
\end{figure}
\end{comment}

\begin{example}\label{exam:p_21}\rm{
Take $n=2$ in Theorems~\ref{thm:a2D2n2n} and ~\ref{thm:b2D2n2n}.
Then $\lpf =\lpf_{4,4}^{2D}$ in \eqref{a2D2n2n} with $n=2$ is given by
\[
\wh{\lpf}(\omega_1,\omega_2)=\frac 1{32}\Big(9+9e^{-i\go_1}+9e^{-i\go_2}+9e^{-i(\go_1+\go_2)}-e^{i(\go_1+\go_2)}-e^{i(\go_2-2\go_1)}
-e^{i(\go_1-2\go_2)}-e^{-i2(\go_1+\go_2)}\Big).
\]
%Then $a$ is $D_4$-symmetric with $\sr(a,M_{\sqrt{2}})=4$, $\lpm(a)=4$, and $\sm(a,M_{\sqrt{2}})=3.0365$.
By $\wh{\hpf_1}(\gbo):=e^{-i\go_1}\ol{\wh{\lpf}(\gbo+(\pi,\pi))}$, the filters $\lpf$ and $\hpf_1$ are given by
\[
a=\frac{1}{32}\left[ \begin{matrix} -1&0&0&-1\\ 0&9&9&0
\\ 0&\boxed{9}&9&0\\ -1&0&0&-1\end{matrix}
 \right]_{[-1,2]^2},\qquad
\hpf_1=\frac{1}{32}\left[ \begin{matrix} 1&0&0&-1\\ 0&\boxed{-9}&9&0
\\ 0&9&-9&0\\ -1&0&0&1\end{matrix}
 \right]_{[-1,2]\times[-2,1]}.
\]
Let $\lpf^D_2$ be the Daubechies orthonormal $2$-wavelet filter given by
\begin{equation}
\label{def_Daub2}
\wh{\lpf^D_2}(\go)=\frac 18\Big((1-\sqrt 3)e^{i\go}+(3-\sqrt 3)+(3+\sqrt 3)e^{-i\go}+(1+\sqrt 3)e^{-i2\go}\Big).
\end{equation}
Define $\wh{v}(\go):=2\wh{\lpf^D_2}(\go/2)\wh{\lpf^D_2}(\go/2+\pi)$.
Then $\wh{\hpf_2}(\go_1,\go_2) := \frac{1}{2}(\wh{v}(\go_1+\go_2)+\wh{v}(\go_1-\go_2)e^{-i\go_2})$ is given by
\[
\hpf_2=\frac{1}{32}\left[ \begin{matrix} \sqrt{3}-2&0&0&2+\sqrt {3}
\\ 0&-\sqrt {3}+6&-\sqrt {3}-6&0
\\ 0&\boxed{-\sqrt {3}+6}&-\sqrt {3}-6&0
\\ \sqrt {3}-2&0&0&2+\sqrt {3}\end{matrix} \right]_{[-1,2]^2}.
\]
By $\wh{\hpf_3}(\gbo):=e^{-i\go_1}\ol{\wh{\hpf_2}(\gbo+(\pi,\pi))}$, the filter $\hpf_3$ is given by
\[
\hpf_3=\frac{1}{32}\left[ \begin{matrix} -2-\sqrt {3}&0&0&\sqrt {3}-2
\\ 0&\sqrt {3}+6&\boxed{-\sqrt {3}+6}&0\\ 0
&-\sqrt {3}-6&\sqrt {3}-6&0\\ 2+\sqrt {3}&0&0&2-
\sqrt {3}\end{matrix} \right]_{[-1,2]\times[-2,1]}.
\]
Note that $\sr(\lpf,M_{\sqrt{2}})=4$, $\lpm(a)=4$, and $\sm(\lpf,M_{\sqrt{2}})=\sm(\lpf,N_{\sqrt{2}})\approx 3.03654$. Hence $\phi^{M_{\sqrt{2}}}, \phi^{N_{\sqrt{2}}} \in C^2(\R^2)$.
The filter $\lpf$ is $D_4$-symmetric about $(1/2,1/2)$, while $\hpf_1$ has the symmetry in \eqref{sym:b1} and $\hpf_2,\hpf_3$ have the symmetry in \eqref{sym:b2b3} with
$\vmo(\hpf_1)=4$ and $\vmo(\hpf_2)=\vmo(\hpf_3)=2$.
The filter bank $\{\lpf; \hpf_1, \hpf_2,\hpf_3\}$ is a double canonical quincunx tight frame filter bank.
Let $\phi,\psi_1,\psi_2,\psi_3$ be defined in \eqref{phi:psi} with $M=M_{\sqrt{2}}, L=3$ and $a=a^{2D}_{4,4}$. Then $\{\phi;\psi_1,\psi_2,\psi_3\}$ is a tight $M_{\sqrt{2}}$-framelet in $L_2(\R^2)$ such that all $\phi,\psi_1,\psi_2,\psi_3$ all have symmetry property as in \eqref{phi:sym}, \eqref{psi1:sym}, and \eqref{psi2psi3:sym}.
%See Figure~\ref{fig:2} for the graphs of the generators $\phi,\psi_1,\psi_2,\psi_3$ in this tight $M_{\sqrt{2}}$-framelet $\{\phi;\psi_1,\psi_2,\psi_3\}$.
}\end{example}

\begin{comment}
\begin{figure}
\centering
\begin{subfigure}[b]{0.35\textwidth}
\includegraphics[width=\textwidth]{phi}
\caption{$\phi$}
\label{fig1:phi}
\end{subfigure}%
        %
\begin{subfigure}[b]{0.35\textwidth}
\includegraphics[width=\textwidth]{psi1}
\caption{$\psi_1$}
\label{fig1:psi1}
\end{subfigure}

\begin{subfigure}[b]{0.35\textwidth}
\includegraphics[width=\textwidth]{psi2}
\caption{$\psi_2$}
\label{fig1:psi2}
\end{subfigure}
%
\begin{subfigure}[b]{0.35\textwidth}
\includegraphics[width=\textwidth]{psi3}
\caption{$\psi_3$}
\label{fig1:psi3}
\end{subfigure}
\caption{Graphs of the generators $\phi,\psi_1,\psi_2,\psi_3$ in the tight $M_{\sqrt{2}}$-framelet for $L_2(\R^2)$ in Example~\ref{exam:p_21}.}
\label{fig:2}
\end{figure}
\end{comment}

\section{Double canonical symmetric quincunx tight framelets derived from one-dimensional filters}
\label{sec:double-canonical:2}

Motivated by the special form in \eqref{a2D2n2n} for the two-dimensional quincunx low-pass filters $\lpf^{2D}_{2n,2n}$,
we now further generalize the construction and results in Section~\ref{sec:double-canonical} for building double canonical symmetric quincunx tight framelets from one-dimensional filters.

\begin{theo}\label{thm:a2D}
Let $u\in l_0(\Z)$ be a one-dimensional finitely supported filter with $\wh{u}(0)=1$. Define a two-dimensional filter $a^{2D}$ by
\begin{equation} \label{a2D}
\wh{\lpf^{2D}}(\omega_1,\omega_2)=\frac{1}{2}
[\wh{u}(\omega_1+\omega_2)+\wh{u}(\omega_1-\omega_2)e^{-i\omega_2}].
\end{equation}
Then
\begin{enumerate}
\item[{\rm (i)}] $a^{2D}$ has order $n$ sum rules with respect to $M_{\sqrt{2}}$ if and only if $u$ has $n$ linear-phase moments with phase $1/2$, i.e.,
\begin{equation}\label{u:lpm}
\wh{u}(\go)=e^{-i\go/2}+\bo(|\go|^n),\qquad \go\to 0.
\end{equation}
\item[{\rm (ii)}] $a^{2D}$ has order $n$ linear-phase moments with phase $(1/2,1/2)$ if and only if $u$ has $n$ linear-phase moments with phase $1/2$, i.e., \eqref{u:lpm} holds.
\item[{\rm (iii)}] $a^{2D}$ is $D_4$-symmetric about the point $(1/2,1/2)$ if and only if $u$ is symmetric about the point $1/2$, that is, $u(1-k)=u(k)$ for all $k\in \Z$.
\end{enumerate}
\end{theo}

\begin{proof} The claim in item (iii) can be directly checked. We now prove items (i) and (ii).
If \eqref{u:lpm} holds, then
\[
\begin{aligned}
\wh{\lpf^{2D}}(\gbo+(\pi,\pi))&=\frac{1}{2}(\wh{u}(\go_1+\go_2)-\wh{u}(\go_1-\go_2)e^{-i\go_2})
\\&=
\frac{1}{2} (e^{-i(\go_1+\go_2)/2}-e^{-i(\go_1-\go_2)/2}e^{-i\go_2})+\bo(\|\gbo\|^n)
\\&=\bo(\|\gbo\|^n)
\end{aligned}
\]
as $\gbo\to 0$. Hence, \eqref{u:lpm} implies that $\lpf^{2D}$ has order $n$ sum rules with respect to $M_{\sqrt{2}}$.

Conversely, suppose that $\lpf^{2D}$ has order $n$ sum rules with respect to $M_{\sqrt{2}}$. Then
\[
\frac{1}{2}(\wh{u}(\go_1+\go_2)-\wh{u}(\go_1-\go_2)e^{-i\go_2})=
\wh{\lpf^{2D}}(\gbo+(\pi,\pi))=\bo(\|\gbo\|^n),\qquad \gbo\to 0,
\]
from which we deduce that
\begin{equation}\label{v:lpm}
\wh{v}(\go_1+\go_2)=\wh{v}(\go_1-\go_2)+\bo(\|\gbo\|^n),\quad \gbo\to 0\quad \mbox{with}\quad
\wh{v}(\go):=\wh{u}(\go)e^{i\go/2}.
\end{equation}
By $\wh{u}(0)=1$, we have $\wh{v}(0)=1$. Now \eqref{v:lpm} implies
\[
\wh{v}^{(j)}(0)=\partial_2^j[\wh{v}(\go_1+\go_2)]|_{\go_1=0,\go_2=0}=
\partial_2^j[\wh{v}(\go_1-\go_2)]|_{\go_1=0,\go_2=0}=(-1)^j \wh{v}^{(j)}(0),\qquad \forall\, 0\le j\le n-1
\]
and
\[
\wh{v}^{(j+1)}(0)=\partial_1 \partial_2^j[\wh{v}(\go_1+\go_2)]|_{\go_1=0,\go_2=0}=
\partial_1 \partial_2^j[\wh{v}(\go_1-\go_2)]|_{\go_1=0,\go_2=0}=(-1)^j \wh{v}^{(j+1)}(0),\qquad \forall\, 0\le j\le n-2.
\]
From the above identities, it is easy to deduce that we must have $\wh{v}(0)=1$ and $\wh{v}^{(j)}(0)=0$ for all $j=1,\ldots,n-1$. That is, $\wh{v}(\go)=1+\bo(|\go|^n)$ as $\go\to0$. Consequently, by $\wh{v}(\go)=\wh{u}(\go)e^{i\go/2}$, \eqref{u:lpm} must hold. This proves item (i).

Similarly, if \eqref{u:lpm} holds, then
\[
\begin{aligned}
\wh{\lpf^{2D}}(\gbo)& =\frac{1}{2}(\wh{u}(\go_1+\go_2)+\wh{u}(\go_1-\go_2)e^{-i\go_2})
\\&=
\frac{1}{2} (e^{-i(\go_1+\go_2)/2}+e^{-i(\go_1-\go_2)/2}e^{-i\go_2})+\bo(\|\gbo\|^n)
\\&=
e^{-i(\go_1+\go_2)/2}+\bo(\|\gbo\|^n)
\end{aligned}
\]
as $\gbo\to 0$. Hence, \eqref{u:lpm} implies that $\lpf^{2D}$ has order $n$ linear-phase moments with phase $(1/2,1/2)$. Conversely, if $\lpf^{2D}$ has order $n$ linear-phase moments with phase $(1/2,1/2)$, then we must have
\[
\wh{v}(\go_1+\go_2)=-\wh{v}(\go_1-\go_2)+\bo(\|\gbo\|^n),\quad \gbo\to 0\quad \mbox{with}\quad
\wh{v}(\go)=\wh{u}(\go)e^{i\go/2}.
\]
A similar proof as in the proof of item (i) shows that \eqref{u:lpm} must hold. This proves item (ii).
\end{proof}

For the filter $u$ in Theorem~\ref{thm:a2D}, we also have the following result.

\begin{pro}\label{prop:u:lpm}
For  a finitely supported filter $u\in l_0(\Z)$  with $\wh u(0)=1$ such that $u$ is symmetric about the point $1/2$, $\lpm(u)$ must be an even integer. Moreover, the filter $u$ is symmetric about the point $1/2$ and $u$ has $2n$ linear-phase moments with phase $1/2$ if and only if $u$  takes the following form
\begin{equation}\label{u}
\wh{u}(\go)=2^{-1}(1+e^{-i\go})\Big( \sin^{2n}(\go/2) R(\sin^2(\go/2))+1+\sum_{j=1}^{n-1} \frac{(2j-1)!!}{(2j)!!} \sin^{2j}(\go/2)\Big)
\end{equation}
for some polynomial $R$, where $(2j-1)!!=(2j-1)(2j-3)\cdots(3)(1)$ and $(2j)!!=(2j)(2j-2)\cdots(2)$. In particular, the two-dimensional filter $\lpf^{2D}$ defined in \eqref{a2D} using the filter $u$ in \eqref{u} with $R=0$ is the same filter $\lpf^{2D}_{2n,2n}$ in \eqref{a2D2n2n}.
\end{pro}

\begin{proof}
Note that $u$ is symmetric about the point $1/2$ if and only if $\wh{u}(\go)=e^{-i\go}\wh{u}(-\go)$, that is, $e^{i\go/2}\wh{u}(\go)=e^{-i\go/2} \wh{u}(-\go)$. Moreover, the symmetry of $u$ also implies that $\sum_{k\in \Z} u(k)k=1/2$.
Thus, it is trivial to see that $[e^{i\go/2}\wh{u}(\go)]^{(2j-1)}(0)=0$ for all $j\in \N$. Consequently, by the definition of linear-phase moments with phase $1/2$, $\lpm(u)$ must be an even integer.

Since $u$ is symmetric about $1/2$, we must have $\wh{u}(\go)=2^{-1}(1+e^{-i\go})P(\sin^2(\go/2))$ for some polynomial $P$. Therefore,
$e^{i\go/2}\wh{u}(\go)=\cos (\go/2) P(\sin^2(\go/2))$. Now $u$ has order $2n$ linear-phase moments with phase $1/2$ if and only if
$\cos (\go/2) P(\sin^2(\go/2))=e^{i\go/2}\wh{u}(\go)=1+\bo(|\go|^{2n})$ as $\go \to 0$,
which, by considering $x=\sin^2(\omega/2)$, is further equivalent to
$P(x)=(1-x)^{-1/2}+\bo(x^n)$ as $x\to 0$.
Considering the Taylor expansion of $(1-x)^{-1/2}$ at $x=0$,
we must have
\[
P(x)=x^n R(x)+\sum_{j=0}^{n-1} \binom{-1/2}{j} (-x)^j=x^n R(x)+1+\sum_{j=1}^{n-1} \frac{(2j-1)!!}{(2j)!!} x^j,
\]
for some polynomial $R$.

When $R=0$, the filter $u$ in \eqref{u} is supported inside $[1-n, n]$. Define $\wh{v}(\go):=(\wh{\lpf^I_{2n}}(\go/2)-\wh{\lpf^I_{2n}}(\go/2+\pi))e^{-i\go/2}$.
Since $\lpf^I_{2n}$ is an interpolatory $2$-wavelet filter, it is trivial to see that $\wh{v}(\go)=(1-2\wh{\lpf^I_{2n}}(\go/2+\pi))e^{-i\go/2}$.
By $\sr(\lpf^I_{2n},2)=2n$,
it is trivial to see that
\[
\wh{v}(\go)=e^{-i\go/2}+\bo(|\go|^{2n}),\qquad \go\to 0.
\]
That is, $\lpm(v)\ge 2n$. Since $\lpf^I_{2n}$ is supported inside $[1-2n,2n-1]$, we deduce that $v$ is supported inside $[1-n,n]$. By the uniqueness of $u$, we must have $v=u$.
This proves $\lpf^{2D}=\lpf^{2D}_{2n,2n}$ in \eqref{a2D2n2n}.
\end{proof}

We now construct double canonical quincunx tight framelet filter banks from the low-pass filters in \eqref{a2D}.

\begin{theo}\label{thm:general}
Let $u\in l_0(\Z)$ be a finitely supported filter such that
\begin{equation}\label{u:positive}
|\wh{u}(\go)|\le 1, \qquad \go\in \R.
\end{equation}
Define $\lpf^{2D}$ as in \eqref{a2D}, $\wh{\hpf_1}(\gbo):= e^{-i\go_1}\overline{\wh{\lpf^{2D}}(\gbo+(\pi,\pi))}$, and
\begin{equation} \label{def:b2b30}
\wh{\hpf_2}(\gbo):=\frac{1}{2}[\wh{v}(\go_1+\go_2)+\wh{v}(\go_1-\go_2)e^{-i\go_2}],
\quad \wh{\hpf_3}(\gbo) :=e^{-i\go_1}\overline{\wh{\hpf_2}(\gbo+(\pi,\pi))},
\end{equation}
where $v\in l_0(\Z)$ is a filter obtained from Fej\'er-Riesz lemma and satisfying $|\wh{v}(\go)|^2=1-|\wh{u}(\go)|^2$.
Then $\{\lpf^{2D};{\hpf_1}, {\hpf_2}, {\hpf_3}\}$ is a double canonical quincunx tight framelet filter bank.
\end{theo}

\begin{proof} By the definitions of $a=\lpf^{2D}$ in \eqref{a2D} and $\hpf_2$ in \eqref{def:b2b30}, as proved in the proof of Theorem~\ref{thm:b2D2n2n}, \eqref{a:square} and \eqref{b:square} must hold. Since $|\wh{u}(\go)|^2+|\wh{v}(\go)|^2=1$,
it is trivial to see that \eqref{sumofsquare} holds with $s=2$. Hence,
$\{\lpf^{2D};{\hpf_1}, {\hpf_2}, {\hpf_3}\}$ is a double canonical quincunx tight framelet filter bank.
\end{proof}

By the same proof of Theorem~\ref{thm:general}, we have the following generalized result of Theorem~\ref{thm:general}:

\begin{theo}
Let $u,v\in l_0(\Z)$ be finitely supported filters such that
\begin{equation}\label{u:v:1}
|\wh{u}(\go)|^2+|\wh{v}(\go)|^2=1, \qquad \go\in \R.
\end{equation}
Let $M$ be a $d\times d$ integer matrix such that $|\det(M)|=2$. Define
\[
\wh{\lpf^{dD}}(\gbo):=\frac{1}{2}[\wh{u}(\gamma_1\cdot \gbo)+\wh{u}(\gamma_2\cdot \gbo)e^{-i\gamma_3\cdot \gbo}],\quad
\wh{\hpf_1}(\gbo):= e^{-i\gamma_4\cdot\gbo}\overline{\wh{\lpf^{dD}}(\gbo+2\pi\xi)},
\]
and
\begin{equation} \label{def:b2b3}
\wh{\hpf_2}(\gbo):=\frac{1}{2}[\wh{v}(\gamma_1\cdot \gbo)+\wh{v}(\gamma_2\cdot\gbo)e^{-i\gamma_3\cdot \gbo}],
\quad \wh{\hpf_3}(\gbo) :=e^{-i\gamma_4\cdot\gbo}\overline{\wh{\hpf_2}(\gbo+2\pi\xi)},
\end{equation}
where $\gbo\in\R^d$, $\xi\in\Omega_{M}\backslash\{0\}$, $\gamma_1,\gamma_2\in M\Z^d\backslash\{0\}$, and $\gamma_3,\gamma_4\in \Z^d\backslash [M\Z^d]$.
Then $\{\lpf^{dD};{\hpf_1}, {\hpf_2}, {\hpf_3}\}$ is a double canonical tight $M$-framelet filter bank.
\end{theo}

For a real-valued symmetric filter $u$ satisfying
\begin{equation}\label{u:sym}
u(1-k)=u(k),\qquad \mbox{for all}\;\, k\in \Z,
\end{equation}
it is of interest to ask whether there exists a finitely supported real-valued filter $v$ satisfying \eqref{u:v:1} with symmetry so that the constructed high-pass filters $\hpf_2$ and $\hpf_3$ in Theorem~\ref{thm:general}
will have better symmetry as in Example~\ref{exam:p_20}. This is negatively answered by the following result.

\begin{theo}\label{thm:uv:sym}
Let $u,v\in l_0(\Z)$ be two finitely supported real-valued filters.
Then \eqref{u:v:1} and \eqref{u:sym} hold, $\sum_{k\in \Z} u(k)=1$, and $v$ has symmetry if and only if
\begin{equation}\label{uv:Haar}
\wh{u}(\go)= 2^{-1}(e^{ij\go}+e^{-i(j+1)\go})\quad \mbox{ and }\quad \wh{v}(\go)=2^{-1}e^{-ik\go}(e^{ij\go}-e^{-i(j+1)\go})
\end{equation}
for some $j,k\in \Z$.
\end{theo}

\begin{proof} The sufficient part is trivial, since \eqref{uv:Haar} implies \eqref{u:v:1} and $v$ has symmetry.

We now prove the necessity part.
Since $u$ has symmetry in \eqref{u:sym}, we can write $\wh{u}(\go)=2^{-1}(1+e^{-i\go})P(\sin^2(\go/2))$ for some polynomial $P$ with real coefficients. Since $\wh{u}(0)=1$, we must have $P(0)=1$. Consequently, we have
$|\wh{u}(\go)|^2=\cos^2(\go/2)(P(\sin^2(\go/2)))^2=(1-x)(P(x))^2$ with $x:=\sin^2(\go/2)$.

Since $v$ has symmetry and there are essentially four different types of symmetry, we must have
\begin{equation}\label{v:sym:0}
\begin{split}
&\wh{v}(\go)=e^{-ik\go} Q(\sin^2(\go/2)),\\
&\wh{v}(\go)=e^{-ik\go}2^{-1}(1+e^{-i\go}) Q(\sin^2(\go/2)),\quad\\
&\wh{v}(\go)=e^{-ik\go} 2^{-1}(e^{i\go}-e^{-i\go})Q(\sin^2(\go/2))
\end{split}
\end{equation}
or
\begin{equation}\label{v:sym}
\wh{v}(\go)=e^{-ik\go}2^{-1}(1-e^{-i\go}) Q(\sin^2(\go/2))\quad\;
\end{equation}
for some $k\in \Z$ and some polynomial $Q$ with real coefficients.
We now show that $v$ must have the symmetry in \eqref{v:sym}.
Otherwise,
$v$ must take one of the three forms in \eqref{v:sym:0}. Then $|\wh{v}(\go)|^2=(Q(x))^2$, $(1-x)(Q(x))^2$, or $x(1-x)(Q(x))^2$, respectively.
Now by $|\wh{u}(\go)|^2+|\wh{v}(\go)|^2=1$, we will have
\[
(1-x)(P(x))^2+(Q(x))^2=1,\quad (1-x)(P(x))^2+(1-x)(Q(x))^2=1,
\quad \mbox{or}\quad (1-x)(P(x))^2+x(1-x)(Q(x))^2=1.
\]
The last two identities cannot hold due to the factor $1-x$, while the first identity must fail by considering $x\to -\infty$ and noting $P\not \equiv 0$.
Thus, $v$ must have the symmetry in \eqref{v:sym}.

By \eqref{u:sym} and \eqref{v:sym}, we see that both $e^{i\omega/2}\wh{u}(\omega)$
and $ie^{i(k+1/2)\omega}\wh{v}(\omega)$ are real-valued. Therefore,
\begin{align*}
e^{i\omega}\big[\wh{u}(\omega)+e^{ik\omega}\wh{v}(\omega)\big]
\big[\wh{u}(\omega)-e^{ik\omega}\wh{v}(\omega)\big]
&=\big[e^{i\omega/2}\wh{u}(\omega)+e^{i(k+1/2)\omega}\wh{v}(\omega)\big]
\big[e^{i\omega/2}\wh{u}(\omega)-e^{i(k+1/2)\omega}\wh{v}(\omega)\big]\\
&=|\wh{u}(\omega)|^2+|\wh{v}(\omega)|^2=1.
\end{align*}
Hence, the first two nontrivial factors in the above identities must be monomial, that is,
\[
\wh{u}(\omega)+e^{ik\omega}\wh{v}(\omega)=\lambda e^{ij\omega},
\quad
\wh{u}(\omega)-e^{ik\omega}\wh{v}(\omega)=e^{-i(j+1)\omega}/\lambda
\]
for some $j\in \Z$ and $\lambda \in \R\backslash \{0\}$. From the above identities, we have $\wh{u}(\omega)=[\lambda e^{ij\omega}+e^{-i(j+1)\omega}/\lambda]/2$.
By \eqref{u:sym}, we must have $\lambda=1$ and \eqref{uv:Haar} holds.
\end{proof}

By Theorem~\ref{thm:uv:sym}, we can conclude that except the Haar-type double canonical quincunx tight framelet filter bank that is similar to Example~\ref{exam:p_20}, there is no other real-valued double canonical quincunx tight framelet filter bank with better symmetry property. Moreover, due to Proposition~\ref{prop:u:lpm}, it is quite easy to observe that the real-valued low-pass filter $u$ constructed in \eqref{uv:Haar} can have no more than two linear-phase moments and therefore, the tight framelet filter banks constructed in Theorem~\ref{thm:general} can have no more than one vanishing moment. This shortcoming can be easily remedied by using complex-valued filters. As shown in \cite[Theorem~1 and Algorithm~2]{Han:adv:2010}, there are finitely supported complex-valued low-pass orthonormal $2$-wavelet filters $\lpf$ such that $\lpf(1-k)=\lpf(k)$ for all $k\in \Z$ with arbitrarily high orders of sum rules and linear-phase moments. Take $u=\lpf$. Then we can easily obtain complex-valued double canonical symmetric quincunx tight framelet filter banks with arbitrarily high orders of vanishing moments. For the convenience of the reader, we provide an example here by combining \cite[Algorithm~1]{Han:adv:2010} and Proposition~\ref{prop:u:lpm}.

Note that $M_{\sqrt{2}}$ is not compatible with the symmetry group
$D_4^+:=\{\pm \mbox{diag}(1, 1), \pm \mbox{diag}(1,-1)\}$. But we have
$M_{\sqrt{2}}D_4^+M_{\sqrt{2}}^{-1}:=\{M_{\sqrt{2}}EM_{\sqrt{2}}^{-1}: E\in D_4^+\} = \{\pm
\left[\begin{matrix} 1&0\\ 0&1\end{matrix}\right], \pm \left[\begin{matrix}0 & 1\\1 & 0\end{matrix}\right]\} =:D_4^-$ and $M_{\sqrt{2}}D_4^-M_{\sqrt{2}}^{-1}=D_4^+$.

\begin{example}
\label{exam:p_30}
{\rm
Take $n=3$ and $R=0$ in \eqref{u} of Proposition~\ref{prop:u:lpm}. Then $P(x)=1+\frac{1}{2}x+\frac{3}{8}x^2$  and
\[
\tilde{Q}(x):=\frac{1-(1-x)(P(x))^2}{x}=\frac{x^2(9x^2+15x+40)}{64}\ge 0, \qquad \forall\; x\in \R.
\]
Then $Q(x):=\frac38x(x+\frac{5+i3\sqrt{15}}{6})$ satisfies $|Q(x)|^2=\tilde{Q}(x)$ for all $x\in \R$. Define filters $u$ and $v$ by
\[
\wh{u}(\go)=2^{-1}(1+e^{-i\go})P(\sin^2(\go/2)),\qquad
\wh{v}(\go)=2^{-1}(1-e^{-i\go})Q(\sin^2(\go/2)).
\]
Then $\lpm(u)=6$ with phase $c=1/2$,
\[
\wh{u}(\go)  =\frac{1}{256}(150(1+e^{-i\go})-25(e^{i\go}+e^{-2i\go})+3(e^{2i\go}+e^{-3i\go})),
\]
and
\[
\wh{v}(\go)  =\frac{1}{256}((60+i18\sqrt{15})(1-e^{-i\go})-(25+i6\sqrt{15})(e^{i\go}-e^{-2i\go})+3(e^{2i\go}-e^{-3i\go})).
\]
The filters $u$ and $v$ satisfy $u(k)=u(1-k)$ and $v(k)=-v(1-k)$ for $k\in\Z$; that is, $u$ is symmetric about $1/2$ while $v$ is antisymmetric about $1/2$.
Note that the real-valued filter $v$ in Theorem~\ref{thm:b2D2n2n} defined by $\wh{v}(\go) = 2\wh{a_n^D}(\go/2)\wh{a_n^D}(\go/2+\pi)$ does not have symmetry property.

Define $a = a^{2D}$ as in \eqref{a2D}. Then, the filter $a$ satisfies $a=a^{2D}=a_{6,6}^{2D}$ in \eqref{a2D2n2n} due to $P(x)= (1-x)^{-1/2}+\bo(x^4)$ and $a$ is supported on $[-2,3]^2$.
The canonical high-pass filter $b_1$ of $a$ is given by $\wh{b_1}(\gbo)=e^{-i\go_1}\ol{\wh{a}(\gbo+(\pi,\pi))}$. The filters $\lpf$ and $\hpf_1$ are given by
\[
a=\frac{1}{512}\left[
\begin{matrix} 3&0&0&0&0&3\\ 0&-25&0 &0&-25&0\\ 0&0&150&150&0&0\\
0&0& \boxed{150}&150&0&0\\ 0&-25&0&0&-25&0\\ 3&0
&0&0&0&3\end{matrix} \right]_{[-2,3]^2},\qquad
b_1=\frac{1}{512}\left[ \begin{matrix} -3&0&0&0&0&3\\ 0&25&0 &0&-25&0\\ 0&0&\boxed{-150}&150&0&0\\ 0&0& 150&-150&0&0\\ 0&-25&0&0&25&0\\ 3&0 &0&0&0&-3\end{matrix} \right]_{[-2,3]\times[-3,2]}.
\]
Note that $a=a_{6,6}^{2D}$ is real-valued and $D_4$-symmetric about $\mathbf{c}=(1/2,1/2)$ while $b_1$ has symmetry given by \eqref{sym:b1}.
Define high-pass filters $b_2$ and $b_3$ by \eqref{def:b2b30}. Then, the high-pass filter $b_2$ is supported on $[-2,3]^2$ and is given by
\[
b_2= \frac{1}{512}\left[ \begin {array}{cccccc} 3&0&0&0&0&-3\\ \noalign{\medskip}0&
-25-6\,
i\sqrt {15}&0&0&25+6\,i\sqrt {15}&0\\ \noalign{\medskip}0&0&60+18\,
i\sqrt {15}&-60-18\,i\sqrt {15}&0&0\\ \noalign{\medskip}0&0&\boxed{60+18\,i
\sqrt {15}}&-60-18\,i\sqrt {15}&0&0\\ \noalign{\medskip}0&-25-6\,i\sqrt {
15}&0&0&6\,25+i\sqrt {15}&0\\ \noalign{\medskip}3&0&0&0&0&-3
\end {array} \right]_{[-2,3]^2}.
\]
The canonical high-pass filter $b_3$ of $b_2$ is supported on $[-2,3]\times[-3,2]$ and is given by
\[
b_3=\frac{1}{512}\left[ \begin {array}{cccccc} 3&0&0&0&0&3\\ \noalign{\medskip}0&6\,i
\sqrt {15}-25&0&0&6\,i\sqrt {15}-25&0\\ \noalign{\medskip}0&0&\boxed{60-18\,i
\sqrt {15}}&60-18\,i\sqrt {15}&0&0\\ \noalign{\medskip}0&0&-60+18\,
i\sqrt {15}&-60+18\,i\sqrt {15}&0&0\\ \noalign{\medskip}0&25-6\,i
\sqrt {15}&0&0&25-6\,i\sqrt {15}&0\\ \noalign{\medskip}-3&0&0&0&0&-3
\end {array} \right]_{[-2,3]\times[-3,2]}.
\]
The high-pass filters $b_2$ and $b_3$ are complex-valued and have the following symmetry:
\[
\begin{aligned}
b_2(E(k-\mathbf{c})+\mathbf{c}) &= E_{1,1} b_2(k),\quad \forall k\in\Z^2, E\in D_4^+\qquad \mbox{ with }\quad \mathbf{c}=(1/2,1/2),\\
b_3(E(k-\mathring{\mathbf{c}})+\mathring{\mathbf{c}}) & =E_{2,2} b_3(k),\quad \forall k\in\Z^2, E\in D_4^+\qquad \mbox{ with }\quad \mathring{\mathbf{c}}=(1/2,-1/2),
\end{aligned}
\]
where $E_{i,j}$ is the $(i,j)$-entry of $E$.
%and $D_4^+$ is a subgroup of $D_4$ given by
%\[
%D_4^+=\left\{
%\pm\left[\begin{matrix}1 & 0\\ 0 & 1\end{matrix}\right],
%\pm\left[\begin{matrix}1 & 0\\ 0 & -1\end{matrix}\right].
%\right\}.
%\]

The filter bank $\{\lpf; \hpf_1, \hpf_2,\hpf_3\}$ is a double canonical quincunx tight framelet filter bank with
$\vmo(\hpf_1)=6$ and $\vmo(\hpf_2)=\vmo(\hpf_3)=3$.
%See Figure~\ref{fig:3} for the graphs of the generators $\phi,\psi_1,\psi_2,\psi_3$ in the tight $M_{\sqrt{2}}$-framelet $\{\phi;\psi_1,\psi_2,\psi_3\}$,
Let $\phi,\psi_1,\psi_2,\psi_3$ be defined in \eqref{phi:psi} with $M=M_{\sqrt{2}}$, $L=3$ and $a=a^{2D}_{6,6}$. Then $\{\phi;\psi_1,\psi_2,\psi_3\}$ is a tight $M_{\sqrt{2}}$-framelet in $L_2(\R^2)$ such that $\phi,\psi_1$ have symmetry property as in \eqref{phi:sym}, \eqref{psi1:sym}. $\psi_2,\psi_3$ are of complex value %with $\psi_2 = \psi_2^{[r]}+i \psi_2^{[i]}$, $\psi_3 = \psi_3^{[r]}+i\psi_3^{[i]}$,
and have the following symmetry property:
\[
\begin{aligned}
\psi_2(E(\cdot-\mathbf{c}_2)+\mathbf{c}_2)&=[MEM^{-1}]_{1,1}\psi_2,
\\
\psi_3(E(\cdot-\mathbf{c}_3)+\mathbf{c}_3)&=[MEM^{-1}]_{2,2}\psi_3,
\end{aligned}
\qquad \forall E\in D_4^-.
\]
where $\mathbf{c}_2 =(3/2,1/2)$ and $\mathbf{c}_3=(1,1)$.

}\end{example}

\section{Multiple canonical quincunx tight framelet filter banks with symmetry}
\label{sec:multiple-canonical}

\setcounter{equation}{0}

In this section we study symmetric multiple canonical quincunx tight framelet filter banks derived from one-dimensional filters.

As discussed in Section~\ref{sec:intro}, for every $d\times d$ dilation matrix $M$, compactly supported tight $M$-framelets $\{\phi; \psi_1,\ldots, \psi_L\}$ with arbitrarily high vanishing moments and smoothness can be easily constructed (e.g. \cite[Theorem~1.1]{Han:jcam:2003}) but at the cost of large number $L$ of wavelet/framelet functions. The key idea to construct such and similar compactly supported tight $M$-framelets in \cite{GRon:pams:1998,Han:laa:2002,Han:jcam:2003,Han:mmnp:2014,RS_1998} is to use the almost separable low-pass filters in \eqref{separablefilter}. For example,
for two one-dimensional tight $2$-framelet filter banks $\{\hpf_{0}; \hpf_{1}, \ldots, \hpf_{J}\}$ and $\{u_{0}; u_{1}, \ldots, u_{L}\}$ one can trivially verify  (see \cite[Lemma~3.2]{Han:jcam:2003} and \cite{RS_1998}) that
\begin{equation}\label{tensor}
\{b_{j}\otimes u_{k}\; : \; 0\le j \le J, 0\le k\le L\},
\end{equation}
is a quincunx tight framelet filter bank derived from the separable low-pass filter $\hpf_0\otimes u_0$, where $\wh{\hpf_{j}\otimes u_{k}}(\go_1,\go_2):=
\wh{\hpf_j}(\go_1)\wh{u_k}(\go_2)$. Moreover, every one-dimensional tight $2$-framelet filter bank $\{\hpf_{0}; \hpf_{1}, \ldots, \hpf_{J}\}$ is automatically a quincunx tight framelet filter bank by identifying $\Z$ with either $\Z\times \{0\}$ or $\{0\}\times \Z$ so that a one-dimensional filter can be regarded as a two-dimensional filter (\cite{Han:laa:2002}).
Such tight framelet filter banks are particular instances of the tight framelet filter banks constructed via the projection method in \cite{Han:mmnp:2014}. In fact, one can directly check that if $\{\hpf_{0}; \hpf_{1}, \ldots, \hpf_{J}\}$ is a one-dimensional tight $2$-framelet filter bank and if the filters $u_{0}, u_{1}, \ldots, u_{L}$ satisfy
\begin{equation}\label{sos:1}
|\wh{u_0}(\go)|^2+|\wh{u_1}(\go)|^2+\cdots+|\wh{u_L}(\go)|^2=1,
\end{equation}
then the filter bank in \eqref{tensor} is still a quincunx tight framelet filter bank. Note that \eqref{sos:1} is weaker than requiring $\{u_{0}; u_{1}, \ldots, u_{L}\}$ to be a tight $2$-framelet filter bank. For every pair of finitely supported low-pass filters $\hpf_0$ and $u_0$ satisfying $|\wh{\hpf_0}(\go)|^2+|\wh{\hpf_0}(\go+\pi)|^2\le 1$ and $|\wh{u_0}(\go)|^2\le 1$, one can always construct
(\cite{Daubechies2003}) a finitely supported tight $2$-framelet filter bank $\{\hpf_{0}; \hpf_{1}, \hpf_{2}\}$, and by Fej\'er-Riesz lemma, there always exists a finitely supported filter $u_1$ such that \eqref{sos:1} holds with $L=1$. Consequently, the quincunx tight framelet filter bank in \eqref{tensor} with $J=2$ and $L=1$ has only five high-pass filters derived from the given low-pass $b_0\otimes u_0$, and $\{\hpf_{0}; \hpf_{1}, \hpf_{2}\}$ is a quincunx tight framelet filter bank with only two high-pass filters.
However, such quincunx tight framelet filter banks often lack symmetry and are not necessarily a multiple canonical quincunx tight framelet filter bank.
By modifying \eqref{tensor} slightly,
we next show that multiple canonical quincunx tight frame filter banks can be easily obtained from one-dimensional tight framelet filter banks as long as $\{\hpf_{0}; \hpf_{1}, \ldots, \hpf_{J}\}$ has the multiple canonical property.

\begin{theo}\label{thm:multiple-canonical}
Let $s,L$ be positive integers. Suppose that $\{\hpf_{0}; \hpf_{1}, \ldots, \hpf_{2s-1}\}$ is a one-dimensional finitely supported $s$-multiple canonical tight $2$-framelet filter bank having the canonical property: $\wh{\hpf_{2j+1}}(\go) = e^{-i\go}\ol{\wh{\hpf_{2j}}(\go+\pi)}$, $j=0,\ldots,s-1$.
Suppose that $u_0,u_1,\ldots, u_L\in l_0(\Z)$ are one-dimensional filters satisfying \eqref{sos:1}. Then
$\{\hpf^{2D}_{j,k}\, :\, j=0,\ldots,2s-1; k=0,\ldots, L\}$ is an $s(L+1)$-multiple canonical quincunx tight framelet filter bank, where
\begin{equation}\label{multiple:2D}
\wh{\hpf^{2D}_{2j,k}}(\gbo):=\wh{\hpf_{2j}}(\go_1)\wh{u_{k}}(\go_2), \qquad
\wh{\hpf^{2D}_{2j+1,k}}(\gbo):=\wh{\hpf_{2j+1}}(\go_1)\ol{\wh{u_{k}}(\go_2+\pi)},\qquad
\gbo=(\go_1,\go_2)\in\RR^2.
\end{equation}
for $j=0,\ldots, s-1$ and  $k=0,\ldots, L$.
\end{theo}

\begin{proof}
By the canonical property in \eqref{canonical:1} and \eqref{canonical:2}, it follows directly from the definition of $\hpf^{2D}_{j,k}$ that the two-dimensional filter bank $\{\hpf^{2D}_{j,k}\, :\, j=0,\ldots,2s-1; k=0,\ldots, L\}$ has the desired $s(L+1)$-multiple canonical property.
On the other hand, we have
\[
\begin{aligned}
\sum_{j=0}^{2s-1}\sum_{k=0}^{L} |\wh{\hpf^{2D}_{j,k}}(\gbo)|^2
&=
\sum_{j=0}^{s-1}\sum_{k=0}^{L}|\wh{\hpf^{2D}_{2j,k}}(\gbo)|^2
+\sum_{j=0}^{s-1}\sum_{k=0}^{L}|\wh{\hpf^{2D}_{2j+1,k}}(\gbo)|^2\\
&=\sum_{j=0}^{s-1} |\wh{\hpf_{2j}}(\go_1)|^2\sum_{k=0}^L |\wh{u_k}(\go_2)|^2
+\sum_{j=0}^{s-2}|\wh{\hpf_{2j+1}}(\go_1)|^2\sum_{k=0}^L |\wh{u_k}(\go_2+\pi)|^2
=\sum_{\ell=0}^{2s-1}|\wh{\hpf_{\ell}}(\go_1)|^2=1.
\end{aligned}
\]
The fact
$$
\sum_{j=0}^{2s-1}\sum_{k=0}^{L} \wh{\hpf^{2D}_{j,k}}(\gbo)\overline{\wh{\hpf^{2D}_{j,k}}(\gbo+(\pi, \pi))}=0
$$
can be proved similarly. Thus $\{\hpf^{2D}_{j,k}\, :\, j=0,\ldots,2s-1; k=0,\ldots, L\}$ is a quincunx tight framelet filter bank.
\end{proof}

Before applying Theorem~\ref{thm:multiple-canonical} to construct multiple canonical quincunx tight framelets, let us look at the smoothness exponent of the low-pass filter $b_0\otimes u_0$ in Theorem~\ref{thm:multiple-canonical}.

\begin{theo}\label{thm:sm}
Let $1\le p\le \infty$. The following statements hold.
\begin{enumerate}
\item[{\rm (i)}] For $a\in l_0(\Z)$ with $\wh{a}(0)=1$, $\sr(a,M_{\sqrt{2}})=\sr(a,N_{\sqrt{2}})=\sr(a,2)$; if $\sm_p(a,2)\ge 0$, then $\sm_p(a,M_{\sqrt{2}})=\sm_p(a,2)$,
    where $a$ is also regarded as a 2D filter by identifying $\Z$ with $\Z\times \{0\}$ in $\Z^2$.
\item[{\rm (ii)}] For $u,v\in l_0(\Z^d)$ with $\wh{u}(0)=\wh{v}(0)=1$ and for any $d\times d$ dilation matrix $M$, $\sr(u*v,M)\ge \sr(u,M)+\sr(v,M)$ and $\sm(u*v,M)\ge \sm_\infty(u*v,M)\ge \sm(u,M)+\sm(v,M)$, where $\wh{u*v}(\gbo):=\wh{u}(\gbo)\wh{v}(\gbo)$.
\item[{\rm (iii)}] For $u,v\in l_0(\Z)$ with $\wh{u}(0)=\wh{v}(0)=1$, $\sr(u\otimes v,M_{\sqrt{2}})=\sr(u\otimes v,N_{\sqrt{2}})\ge \sr(u,2)+\sr(v,2)$; if $\sm(u,2)\ge 0$ and $\sm(v,2)\ge 0$, then $\sm(u\otimes v,M_{\sqrt{2}})\ge\sm_\infty(u\otimes v,M_{\sqrt{2}}) \ge \sm(u,2)+\sm(v,2)$.
\end{enumerate}
\end{theo}

The proof to Theorem~\ref{thm:sm} is given in Appendix~\ref{sec:proofs}. For item (i), $\sm(a,N_{\sqrt{2}})=\sm(a,2)$ often fails. In fact, as Daubechies showed in \cite{Dau} that $\lim_{n\to \infty} \sm(a^I_{2n},2)=\infty$, while numerical calculation in \cite{CD_1993} observed that $\lim_{n\to \infty} \sm(a^I_{2n},N_{\sqrt{2}})=0$. Moreover, as we shall see in the proof of Theorem~\ref{thm:sm} in Appendix~\ref{sec:proofs}, the condition $\sm_p(a,2)\ge 0$ in item (i) cannot be removed to guarantee $\sm_p(a,M_{\sqrt{2}})=\sm_p(a,2)$.

Let $a_n^D$ with $n\ge 1$ be the Daubechies orthogonal filter with $2n$-nonzero coefficients.
Let $b_0=a_n^D$, $a_0=a_m^D$ and define $b_1$ and $a_1$ by
$$
\wh b_1(\go)=e^{-i\go_1}\overline{\wh{a_n^D}(\go+\pi)}, \qquad  \wh a_1(\go)=\wh{ a_m^D}(\go+\pi), \; \go \in \RR,
$$
%where  $a_n^D$ and $a_m^D$ are the Daubechies orthogonal filters,
then we have double canonical quincunx tight framelet filter banks based on the Daubechies orthogonal filters as summarized in the following corollary.

\begin{cor}\label{theo:double_cano_Daub}
Let $a_n^D$ and $a_m^D$ be the Daubechies orthogonal filters. Define
\[
\begin{aligned}
&\wh{\hpf^{2D}_{0}}(\gbo):=\wh{a_n^D}(\go_1)\wh{a_m^D}(\go_2),
\qquad  &\wh{\hpf^{2D}_{1}}(\gbo):=e^{-i\go_1} \overline{\wh{\hpf^{2D}_{0}}(\gbo+(\pi, \pi))}, \\
&\wh{\hpf^{2D}_{2}}(\gbo):=\wh{a_n^D}(\go_1)\wh{a_m^D}(\go_2+\pi),
 \qquad &\wh{\hpf^{2D}_{3}}(\gbo):=e^{-i\go_1} \overline{\wh{\hpf^{2D}_{2}}(\gbo+(\pi, \pi))}.
\end{aligned}
\]
Then $\{\hpf^{2D}_{0}; \hpf^{2D}_{1}, \hpf^{2D}_{2}, \hpf^{2D}_{3}\}$ is a double canonical quincunx tight framelet filter bank such that $\min(\vmo(\hpf_1^{2D})$,$\vmo(\hpf_2^{2D})$, $\vmo(\hpf_3^{2D}))\ge \min(m, n)$ and
$\sm(b_0^{2D},M_{\sqrt{2}})\ge \sm(\lpf_n^{D},2)+\sm(\lpf_m^{D},2)\to\infty$ as $m+n\to \infty$.
% with $p$ having sum rule order $n+m$ and $\{q^{(1)}, q^{(2)}, q^{(3)}\}$ having vanishing moments
%of order $\min\{n, m\}$.
\end{cor}

%Canonical tight $2I_d$-framelet filter banks have been extensively discussed in \cite{JS_2014}.
The Daubechies orthogonal filter-based double canonical quincunx tight framelet filter bank $\{\hpf^{2D}_{0};  \hpf^{2D}_{1}, \hpf^{2D}_{2}, \hpf^{2D}_{3}\}$ does not have any symmetry.
In this paper we are interested in multiple/double canonical quincunx tight framelet filter banks with symmetry.  We immediately conclude from Theorem~\ref{thm:multiple-canonical} that all nontrivial symmetric real-valued canonical qunicunx tight framelet filter banks of the form in \eqref{multiple:2D} must have multiplicity at least $6$. In fact, if we require both $\{b_0;b_1,\ldots,b_{2s-1}\}$ and $\{u_0;u_1,\ldots,u_L\}$ in Theorem~\ref{thm:multiple-canonical} to be of real-valued filters with symmetry, then $s\ge2$ and $L\ge2$ due to the well-known fact that  except the Haar type filter banks, there is no real-valued symmetric  dyadic orthonormal wavelet filter bank. Consequently,  the  multiplicity of a nontrivial canonical
quincunx tight framelet filter bank with real-valued filters and with symmetry satisfies $s(L+1)\ge6$. That is, $\{b_0;b_1,\ldots,b_{2s-1}\}$ need to be at least double canonical tight $2$-framelet filter bank $\{a,b_1,b_2,b_3\}$ while $\{u_0;u_1,\ldots,u_L\}$ need to be at least $\{u_0;u_1,u_2\}$.

We now discuss double canonical tight $2$-framelet filter bank $\{\lpf; \hpf_1,\hpf_2, \hpf_3\}$ with symmetry satisfying
\begin{equation}\label{tffb:2:1d}
\wh{\hpf_1}(\go)=e^{-i\go} \ol{\wh{\lpf}(\go+\pi)},\qquad
\wh{\hpf_3}(\go)=e^{-i\go} \ol{\wh{\hpf_2}(\go+\pi)}.
\end{equation}
It follows trivially from the above relations in \eqref{tffb:2:1d} that
\[
\ol{\wh{\lpf}(\go)}\wh{\lpf}(\go+\pi)
+\ol{\wh{\hpf_1}(\go)}\wh{\hpf_1}(\go+\pi)=0,\qquad
\ol{\wh{\hpf_2}(\go)}\wh{\hpf_2}(\go+\pi)
+\ol{\wh{\hpf_3}(\go)}\wh{\hpf_3}(\go+\pi)=0.
\]
Consequently, every double canonical tight $2$-framelet filter bank with symmetry is a special case of type I symmetric tight $2$-framelet filter banks $\{\lpf; \hpf_1, \hpf_2, \hpf_3\}$ discussed in \cite{Han:acha:2014}. Moreover, Algorithm~1 in \cite{Han:acha:2014} can be used to find all possible such type I symmetric tight $2$-framelet filter banks $\{\lpf; \hpf_1, \hpf_2, \hpf_3\}$ from any given symmetric low-pass filter. For simplicity, we only discuss real-valued filters here. As a special case of \cite[Algorithm~1]{Han:acha:2014}, the following result constructs all possible double canonical tight $2$-framelet filter banks with symmetry.

\begin{theo}\label{thm:double:2}
Let $\lpf\in l_0(\Z)$ be a real-valued low-pass filter having symmetry and satisfying
\begin{equation}\label{cond:a}
\wh{v}(2\go):=1-|\wh{\lpf}(\go)|^2-|\wh{\lpf}(\go+\pi)|^2\ge0, \qquad \forall\; \go\in \R.
\end{equation}
Define a finitely supported real-valued high-pass filter $\hpf_2$ by either of the following two cases:
\begin{enumerate}
\item Obtain a real-valued filter $u\in l_0(\Z)$ through Fej\'{e}r-Riesz lemma by $|\wh{u}(\go)|^2=\wh{v}(\xi)$ and define
\[
\wh{\hpf_2}(\go):=(\wh{u}(2\go)+\epsilon e^{-i\go c_b} \ol{\wh{u}(2\go)})/2\qquad \mbox{with}\quad
\epsilon \in \{-1,1\} \; \mbox{and $c_b$ being an odd integer}.
\]
\item If in addition multiplicity of any zero inside $(0,1)$ of the Laurent polynomial $\sum_{k\in \Z} v(k) z^k$ is even, then one can always construct finitely supported real-valued filters $u_1, u_2$ with symmetry such that
\[
|\wh{u_1}(\go)|^2+|\wh{u_2}(\go)|^2=\wh{v}(\xi) \quad \mbox{with}\quad \frac{\sym \wh{u_1}(\go)}{\sym \wh{u_2}(\go)}=e^{-i\go},
\]
where $\sym \wh{u_1}(\go):=\frac{\wh{u_1}(\go)}{\wh{u_1}(-\go)}$ records the symmetry type of the filter $u_1$. Define
\[
\wh{\hpf_2(\go)}:=(\wh{u_1}(2\go)+e^{-i\go} \wh{u_2}(\go))/\sqrt{2}.
\]
\end{enumerate}
Define the filters $\hpf_1$ and $\hpf_3$ as in \eqref{tffb:2:1d}.
Then $\{\lpf;\hpf_1,\hpf_2,\hpf_3\}$ is a double canonical tight $2$-framelet filter bank such that all the filters have symmetry.
Moreover, all finitely supported canonical tight $2$-framelet filter banks with symmetry can be obtained by the above procedure.
\end{theo}

The construction of real-value filters $\{u_0;u_1,u_2\}$ satisfying \eqref{sos:1} and having symmetry  has been completely solved in \cite[Theorem~2.7]{Han:acha:2013} and \cite[Lemma~2.4]{HanMo:simaa:2004}.

Now we have the main result in this paper on $6$-multiple canonical quincunx tight framelet filter banks with symmetry and vanishing moments.

\begin{theo}
\label{thm:6multiple}
Let $\lpf\in l_0(\Z)$ be a real-valued low-pass filter satisfying the condition in \eqref{cond:a} such that $\wh{\lpf}(0)=1$ and $\lpf$ has symmetry.
Then we can always construct by Theorem~\ref{thm:double:2} a finitely supported real-valued
double canonical tight $2$-framelet filter bank $\{\lpf; \hpf_1,\hpf_2,\hpf_3\}$ with symmetry and
by \cite[Theorem~2.7]{Han:acha:2013} finitely supported real-valued filters $u_1$ and $u_2$ with symmetry such that
\begin{equation}\label{partitionof1}
|\wh{a}(\go)|^2+|\wh{u_1}(\go)|^2+|\wh{u_2}(\go)|^2=1.
\end{equation}
Define two-dimensional filters $\hpf^{2D}_{\ell,k}$ as in \eqref{multiple:2D} of Theorem~\ref{thm:multiple-canonical} for $\ell=0,\ldots,3$ and $k=0,1,2$ with $\hpf_0:=a$ and $u_0:=a$. Define
$\lpf^{2D}:=\hpf^{2D}_{0,0}$.
Then
\begin{equation}\label{multple:6}
\{\lpf^{2D}; \hpf^{2D}_{1,0},\hpf^{2D}_{2,0},\hpf^{2D}_{3,0},\hpf^{2D}_{0,1},\hpf^{2D}_{1,1},
\hpf^{2D}_{2,1},\hpf^{2D}_{3,1},\hpf^{2D}_{0,2},\hpf^{2D}_{1,2},\hpf^{2D}_{2,2},
\hpf^{2D}_{3,2}\}
\end{equation}
is a $6$-multiple canonical quincunx tight framelet filter bank such that
the real-valued low-pass filter $\lpf^{2D}$ is $D_4$-symmetric, with
\[
\sr(a^{2D},M_{\sqrt{2}})=\sr(a^{2D},N_{\sqrt{2}})\ge 2\sr(a,2) \quad \mbox{and}\quad \sm(a^{2D},M_{\sqrt{2}})=\sm(a^{2D},N_{\sqrt{2}})\ge 2\sm(a,2),
\]
and all the eleven high-pass filters are real-valued and have symmetry with at least order $\min(2\sr(a,2), \lpm(a)/2)$ vanishing moments. In particular, if we take $a=a^I_{2n}$ with $n\in \N$, then we have a $6$-multiple canonical quincunx tight framelet filter bank in \eqref{multple:6} such that
\begin{enumerate}
\item[{\rm (i)}] all the high-pass filters have symmetry and at least order $n$ vanishing moments;
\item[{\rm (ii)}] the low-pass $a^{2D}=a^I_{2n}\otimes a^I_{2n}$ is $D_4$-symmetric such that $\sr(a^I_{2n}\otimes a^I_{2n}, M_{\sqrt{2}})=\sr(a^I_{2n}\otimes a^I_{2n}, N_{\sqrt{2}})\ge 4n$, $\lim_{n\to \infty} \sm(a^I_{2n}\otimes a^I_{2n}, M_{\sqrt{2}})=\infty$, and $\lim_{n\to \infty} \sm(a^I_{2n}\otimes a^I_{2n}, N_{\sqrt{2}})=\infty$;
\item[{\rm (iii)}] the tight $M_{\sqrt{2}}$-framelet (or tight $N_{\sqrt{2}}$-framelet) $\{\phi; \psi_1,\ldots, \psi_L\}$ with $L=11$ in $L_2(\R^2)$ have symmetry and arbitrarily high orders of vanishing moments and smoothness, where $\phi,\psi_1,\ldots,\psi_L$ is defined in \eqref{phi:psi}.
\end{enumerate}
\end{theo}

\begin{proof} Since $a$ is $D_4$-symmetric, by definition of smoothness exponent, we can directly verify that $\sm_p(a^{2D},M_{\sqrt{2}})=\sm_p(a^{2D},N_{\sqrt{2}})$ for all $1\le p\le \infty$ (also see Theorem~\ref{thm:a2D2n2n} and \cite{Han:laa:2002,Han:acha:2004}).
It is known in \cite{Dau} that $\lim_{n\to \infty} \sm(a^I_{2n},2)=\infty$. By Theorem~\ref{thm:sm}, we have $\sm(a^{2D},M_{\sqrt{2}})=\sm(a^I_{2n}\otimes a^I_{2n},M_{\sqrt{2}})\ge 2\sm(a^I_{2n},2)$. Consequently, we have $\lim_{n\to \infty} \sm(a^I_{2n}\otimes a^I_{2n},M_{\sqrt{2}})=\infty$.
All other claims follow from the results and discussion before Theorem~\ref{thm:6multiple}.
\end{proof}

As proved in \cite[Theorem~1 and (2.15)]{Han:adv:2010}, there are finitely supported complex-valued orthonormal $2$-wavelet filters with symmetry and arbitrarily high orders of sum rules. As a consequence, if we relax the constrain on real-valued filters and allow complex-valued filter banks, we can have double canonical quincunx tight framelet filter banks with symmetry of form in \eqref{multiple:2D}.

\begin{cor}
For $n\in \N$, let $\lpf_n\in l_0(\Z)$ be the finitely supported symmetric complex-valued orthonormal $2$-wavelet filter with $\sr(\lpf_n,2)=2n-1$ as constructed in \cite[Theorem~1]{Han:adv:2010}.
Define
\[
\wh{\lpf^{2D}}(\go_1,\go_2):=\wh{\lpf_n}(\go_1)\wh{\lpf_n}(\go_2),
\qquad \wh{\hpf_2}(\go_1,\go_2):=\wh{\lpf_n}(\go_1)\wh{\lpf_n}(\go_2+\pi)
\]
and
\[
\wh{\hpf_1}(\go_1,\go_2):=e^{-i\go_1}\ol{\wh{\lpf^{2D}}(\go_1+\pi,\go_2+\pi)},\qquad
\wh{\hpf_3}(\go_1,\go_2):=e^{-i\go_1} \ol{\wh{\hpf_2}(\go_1+\pi,\go_2+\pi)}.
\]
Then $\{\lpf^{2D};\hpf_1,\hpf_2,\hpf_3\}$ is a double canonical quincunx tight framelet filter bank such that $\lpf^{2D}$ is $D_4$-symmetric, with
\[
\sr(\lpf^{2D},M_{\sqrt{2}})\ge 2n\quad\mbox{and}\quad
\sm(\lpf^{2D},M_{\sqrt{2}})=\sm(\lpf^{2D},N_{\sqrt{2}})\ge 2\sm(a_n,2)\to \infty,\quad \mbox{as} \; n\to \infty,
\]
and all the high-pass filters $\hpf_1,\hpf_2,\hpf_3$ have symmetry and at least order $n$ vanishing moments.
\end{cor}

We conclude this section by presenting two examples of $6$-multiple canonical quincunx real-valued tight framelet filter banks to illustrate the result in Theorem~\ref{thm:6multiple}.

\begin{example}
\label{ex:p40}
\rm{
Consider $a = a_2^I=\{-\frac{1}{32},0,\frac{9}{32},\boxed{\tfrac{1}{2}},\frac{9}{32},0,-\frac{1}{32}\}_{[-3,3]}$
with $\sr(a,2)=4$ and $\lpm(a)=4$.
%\[\wh{a}(\go)=\frac12+\frac{9}{32}(e^{i\go}+e^{-i\go})-\frac1{32}(e^{i3\go}+e^{-i3\go}).\]
Then
\[
1-|\wh{a}(\go)|^2-|\wh{a}(\go+\pi)|^2= -\frac{1}{64}(\cos^3(2x)-9\cos^2(2x)+15\cos(2x)-7)\ge0.
\]
By Fej\'{e}r-Riesz Lemma, we can obtain $u\in l_0(\Z)$ such that $|\wh u(2\go)|^2 = 1-|\wh{a}(\go)|^2-|\wh{a}(\go+\pi)|^2$ as follows.
\[
\wh u(\go) = \frac{\sqrt{2}}{32}e^{i\go}(t_0 +t_1e^{-i\go} +t_2 e^{-i2\go}+t_3 e^{-i3\go}),
\]
where $t_0 = 2-\sqrt{3}$, $t_1=-6+\sqrt{3}$, $t_2=6+\sqrt{3}$, $t_3 = -2-\sqrt{3}$.
Define $b_1,b_2,b_3$ by
\begin{equation}\label{expl:b1b2b3}
\wh{b_1}(\go)=e^{-i\go}\ol{\wh{a}(\go+\pi)},\quad
\wh{b_2}(\go)=(\wh{u}(2\go)+e^{-i\go}\ol{\wh{u}(2\go)})/2,\quad
\wh{b_3}(\go)=e^{-i\go}\ol{\wh{b_2}(\go+\pi)}.
\end{equation}
Then,
\[
\begin{aligned}
\wh{b_1}(\go) & = e^{-i\go}\left(\frac12-\frac{9}{32}(e^{i\go}+e^{-i\go})+\frac1{32}(e^{i3\go}+e^{-i3\go})\right),\\
\wh{b_2}(\go) &=\frac{\sqrt{2}}{64}\left(t_3(e^{i3\go}+e^{-i4\go})+t_0(e^{i2\go}+e^{-i3\go})+t_2(e^{i\go}+e^{-i2\go})+t_1(1+e^{-i\go})\right),\\
\wh{b_3}(\go)& =-\frac{\sqrt{2}}{64}\left(t_3(e^{i3\go}-e^{-i4\go})-t_0(e^{i2\go}-e^{-i3\go})+t_2(e^{i\go}-e^{-i2\go})-t_1(1-e^{-i\go})\right).
\end{aligned}
\]
Note that $b_1$ is symmetric about 1 and supported on $[-2,4]$, $b_2$ is symmetric about 1/2 and supported on $[-3,4]$, $b_3$ is antisymmetric about 1/2 and supported on $[-3,4]$. The filter bank $\{a; b_1, b_2, b_3\}$ forms a double canonical tight $2$-framelet filter bank.
See \cite[Examples~4 and~9]{Han:acha:2014} for other tight $2$-framelet filter banks $\{a;b_1,b_2,b_3\}$ with symmetry derived from the interpolatory low-pass filter $a=a^I_2$.

Next, define $\wh{v}(\go):= 1-|\wh{a}(\go)|^2$.
Then,
\[
%\begin{aligned}
\wh{v}(\xi) = \Big|\frac{1-e^{-i\go}}{2}\Big|^4\Big|\frac{e^{-i\go}+2-\sqrt{3}}{\sqrt{4-2\sqrt{3}}}\Big|^2\Big(\frac{6+3\cos(x)-\cos^3(x)}{4}\Big).
%\\&=\Big|\frac{1-e^{-i\go}}{2}\Big|^4\Big|\frac{e^{-i\go}+2-\sqrt{3}}{\sqrt{4-2\sqrt{3}}}\Big|^2\frac{(\cos(x)-c_0)(\cos(x)+c_1)(\cos(x)+\ol{c_1})}{-4},
%\end{aligned}
\]
%where
%\[c_0=t+1/t,\quad c_1 =  (t+1/t-\sqrt{3}i(t-1/t))/2, \quad t =\sqrt[3]{3+2\sqrt{2}}.\]
By Fej\'er-Riesz lemma,
we can obtain $u_0$ such that $|\wh{u_0}(\go)|^2 = v(\go)$ as follows:
\[
\wh{u_0}(\go) =e^{i3\go} \Big(\frac{1-e^{-i\go}}{2}\Big)^2\Big(\frac{e^{-i\go}+2-\sqrt{3}}{\sqrt{4-2\sqrt{3}}}\Big)
\Big(\frac{e^{-i\go}-r_1}{2\sqrt{2r_1}}\Big)
\frac{(e^{-2i\go}+(r_2+\ol{r_2})e^{-i\go}+|r_2|^2)}{2|r_2|}
\]
where
\[
 r_1:=c_0-\sqrt{c_0^2-1}, \; r_2:=c_1-\sqrt{c_1^2-1},
\]
with
\[
t=(3+2\sqrt{2})^{1/3}, \; c_0:=t+\frac1t,\;
c_1:=\frac{c_0}{2}-\frac{\sqrt{3}}{2} i (t-1/t).
\]
%
%\[r_0=c_0-\sqrt{c_0^2-1},\quad r_1 = c_1-\sqrt{c_1^2-1}.\]
Define $u_1, u_2$ by
\[
\wh{u_1}(\go) = (\wh{u_0}(\go)+e^{-i\go}\ol{\wh{u_0}(\go)})/2,\quad
\wh{u_2}(\go) = (\wh{u_0}(\go)-e^{-i\go}\ol{\wh{u_0}(\go)})/2.
\]
Then, $u_1$ is symmetric about 1/2 with support $[-3,4]\cap\Z$, $u_2$ is antisymmetric about 1/2 with support $[-3,4]$, and  $|\wh a(\go)|^2+|\wh{u_1}(\go)|^2+|\wh{u_2}(\go)|^2 =1$.

Finally, we can define
\[
\{\lpf^{2D}; \hpf^{2D}_{1,0},\hpf^{2D}_{2,0},\hpf^{2D}_{3,0},\hpf^{2D}_{0,1},\hpf^{2D}_{1,1},
\hpf^{2D}_{2,1},\hpf^{2D}_{3,1},\hpf^{2D}_{0,2},\hpf^{2D}_{1,2},\hpf^{2D}_{2,2},
\hpf^{2D}_{3,2}\}
\]
as in Theorem~\ref{thm:6multiple}, which gives a $6$-multiple canonical quincunx tight framelet filter bank. $a^{2D}$ has  at least  order $4$ sum rules and is $D_4$-symmetric about the origin. All the eleven high-pass filters are real-valued and have symmetry with at least order $2$ vanishing moments.
}\end{example}

\begin{example}
\label{ex:p50}
\rm{
Consider $a=\{-\frac{3}{64},\frac{5}{64},\boxed{\tfrac{15}{32}},\frac{15}{32},\frac{5}{64},-\frac{3}{64}\}_{[-2,3]}$
with $\sr(a,2)=3$ and $\lpm(a)=4$.
%\[
%\wh{a}(\go)=\frac{15}{32}(1+e^{-i\go})+\frac{5}{64}(e^{i\go}+e^{-i2\go})-\frac{3}{64}(e^{i2\go}+e^{-i3\go}).
%\]
%The filter $a$ is supported on $[-2,3]$ and
Then
\[
1-|\wh{a}(\go)|^2-|\wh{a}(\go+\pi)|^2= -\frac{15}{256}(1-\cos(2x))^2\ge0.
\]
Then $\wh{u}(\go):= \frac{\sqrt{15}}{32}(2-e^{-i\go}-e^{i\go})$ satisfies $1-|\wh{a}(\go)|^2-|\wh{a}(\go+\pi)|^2 = |\wh u(2\go)|^2$.
Define $b_1,b_2,b_3$ as in \eqref{expl:b1b2b3}.
Then,
\[
\begin{aligned}
\wh{b_1}(\go) & =\frac{15}{32}(-1+e^{-i\go})+\frac{5}{64}(e^{i\go}-e^{-i2\go})+\frac{3}{64}(e^{i2\go}-e^{-i3\go}),\\
\wh{b_2}(\go) &=\frac{\sqrt{15}}{64}\left(2(1+e^{-i\go})-(e^{i\go}+e^{-i2\go})-(e^{i2\go}+e^{-i3\go})\right),\\
\wh{b_3}(\go)& =\frac{\sqrt{15}}{64}\left(2(-1+e^{-i\go})-(e^{i\go}-e^{-i2\go})+(e^{i2\go}-e^{-i3\go})\right).
\end{aligned}
\]
Note that high-pass filter $b_1$ is antisymmetric about $1/2$ and supported on $[-2,3]$, the high-pass filter $b_2$ is symmetric about $1/2$ and supported on $[-2,3]$, and the high-pass filter $b_3$ is antisymmetric about $1/2$ and supported on $[-2,3]$. The filter bank $\{a; b_1, b_2, b_3\}$ forms a double canonical tight $2$-framelet filter bank with $\vmo(b_1)=3$, $\vmo(b_2)=2$, and $\vmo(b_3)=3$. See \cite[Example~3]{Han:acha:2013} for a tight $2$-framelet filter bank $\{a;b_1,b_2\}$ with symmetry derived from the low-pass filter $a$.

Next, define $\wh{v}(\go):= 1-|\wh{a}(\go)|^2=-\frac{1}{128}(\cos(x)-1)^2(9\cos^3(x)+3\cos^2(x)-53\cos(x)-79)$.
Then,
\[
\begin{aligned}
\wh{v}(\xi) &=\Big|\frac{1-e^{-i\go}}{2}\Big|^4\frac{9(\cos(x)-c_0)(\cos(x)-c_1)(\cos(x)-\ol{c_1})}{-32},
\end{aligned}
\]
where
\[
c_0=t_1+t_2-1/9,\; c_1 =  -(t_1+t_2+2/9-\sqrt{3}i(t_1-t_2))/2, \; t_1 =8\sqrt[3]{10}/9, \; t_2 = 2\sqrt[3]{100}/9.
\]
Consequently, we can obtain $u_0$ such that $|\wh{u_0}(\go)|^2 = v(\go)$ as follows.
\[
\wh{u_0}(\go) =3e^{i2\go} \Big(\frac{1-e^{-i\go}}{2}\Big)^2\Big(\frac{e^{-i\go}-r_0}{2\sqrt{r_0}}\Big)\frac{e^{-2i\go}-(r_1+\ol{r_1})e^{-i\go}+|r_1|^2}{8|r_1|},
\]
where $r_0=c_0-\sqrt{c_0^2-1}$ and $r_1 = c_1-\sqrt{c_1^2-1}$.
Define $u_1, u_2$ by
\[
\wh{u_1}(\go) = (\wh{u_0}(\go)+e^{-i\go}\ol{\wh{u_0}(\go)})/2,\quad
\wh{u_2}(\go) = (\wh{u_0}(\go)-e^{-i\go}\ol{\wh{u_0}(\go)})/2.
\]
Then, $u_1$ is symmetric about 1/2 with support $[-2,3]\cap\Z$, $u_2$ is antisymmetric about 1/2 with support $[-2,3]\cap\Z$, and  $|\wh a(\go)|^2+|\wh{u_1}(\go)|^2+|\wh{u_2}(\go)|^2 =1$. We also have $\vmo(u_1)=2$ and $\vmo(u_2) = 3$.

Finally, we can define
\[
\{\lpf^{2D}; \hpf^{2D}_{1,0},\hpf^{2D}_{2,0},\hpf^{2D}_{3,0},\hpf^{2D}_{0,1},\hpf^{2D}_{1,1},
\hpf^{2D}_{2,1},\hpf^{2D}_{3,1},\hpf^{2D}_{0,2},\hpf^{2D}_{1,2},\hpf^{2D}_{2,2},
\hpf^{2D}_{3,2}\}
\]
as in Theorem~\ref{thm:6multiple}, which gives a $6$-multiple canonical quincunx tight frame filter bank.
 $a^{2D}$ has  at least  order 6  sum rule and is  $D_4$-symmetric about the origin. All the eleven high-pass filters are real-valued and have symmetry with at least  $2$ vanishing moments.
}
\end{example}

We remark that other $6$-multiple canonical quincunx tight framelet filter banks with high orders of vanishing moments can be obtained by considering other low-pass filters and following the above procedure.

\appendix
\section{Proofs of Theorems~\ref{thm:a2D2n2n} and~\ref{thm:sm}}
\label{sec:proofs}

\begin{proof}[Proof of Theorem~\ref{thm:a2D2n2n}]
The existence of such a filter $\lpf^{2D}_{2n,2n}$ has been proved in Proposition~\ref{prop:u:lpm}. Let $\lpf$ be such a filter $\lpf^{2D}_{2n,2n}$. We now prove the uniqueness of such a filter $\lpf$ satisfying all the properties in Theorem~\ref{thm:a2D2n2n}.

$\lpf$ having orders $2n$ sum rules  with respect to $M=M_{\sqrt{2}}$ is equivalent to
\begin{equation}\label{def:sr}
\sum_{ k\in\Z^2} \lpf(Mk+(1,0))(M k+(1,0))^\mu = \sum_{ k\in\Z^2} \lpf(M k) (M  k)^\mu \quad \forall |\mu|<2n,
\end{equation}
and  $\lpf$ has order $2n$  linear-phase moments with phase $\mathbf{c}=(1/2,1/2)$ is equivalent to
\begin{equation}\label{def:lpm}
\sum_{ k\in\Z^2} a( k) k^\mu = \mathbf{c}^\mu \quad \forall\, |\mu|<2n,
\end{equation}
where $\mu=(\mu_1,\mu_2)\in\NN_0^2$. It is easily seen that
\eqref{def:sr} and \eqref{def:lpm}  are equivalent to
\begin{equation}
\label{eq:sr-lpm}
\begin{cases}
\sum_{ k\in\Z^2} \lpf(M k+(1,0))(M k+(1,0))^\mu & = \frac12 {\bf c}^\mu\\
\\
\sum_{ k\in\Z^2} \lpf(M k) (M  k)^\mu & = \frac12 {\bf c}^\mu
\end{cases},
\quad|\mu|<2n.
\end{equation}

Define
\[
\begin{aligned}
\Lambda_0&:=\{ k=(k_1,k_2)\in\ZZ^2: k_1+k_2 \mbox{ even},  k\in[-n+1,n]^2\},
\\\Lambda_1&:=\{ k=(k_1,k_2)\in\ZZ^2: k_1-k_2 \mbox{ odd},  k\in[-n+1,n]^2\}.
\end{aligned}
\]
Then, $\#\Lambda_0 = \# \Lambda_1 = 2n^2$, $\Lambda_0\cap \Lambda_1=\emptyset$, and $\Lambda_0\cup E_1 = [-n+1,n]^2\cap \ZZ^2$. Moreover,
$\Lambda_0 = M \ZZ^2\cap [-n+1,n]^2$ and $\Lambda_1 = (M\ZZ^2+(1,0))\cap [-n+1,n]^2$.
On the other hand, consider the  index set
\[
\Gamma_{n}:=\{\mu\in\N_0^2: |\mu|<2n, \mu_2<2n-1\}\setminus\{(0,2j-1): j=1,\ldots,n-1\}.
\]
Then, it is easy to show that $\#\Lambda_0=\# \Lambda_1= \#\Gamma_{n} = 2n^2$. Using these notation and noticing $\Gamma_n$ is a subset of $\{\mu\in\NN_0^2: |\mu|<2n\}$, \eqref{eq:sr-lpm} implies that $\lpf$ must also satisfy  the following conditions.
\begin{equation}
\label{eq:sr-lpm2}
\sum_{ k\in \Lambda_\epsilon} \lpf( k) k^\mu  = \frac12 {\bf c}^\mu,  \quad \mu \in \Gamma_n, \epsilon\in\{0,1\}.
\end{equation}

Note that
\[
\# \big(\Lambda_0\cap \{\bx = (x_1,x_2)\in\RR^2: x_1+x_2 = 2j\}\big) = 4-|2j-1|, j = -n+1,\cdots, n
\]
and
\[
\#\big(\Lambda_1\cap \{\bx= (x_1,x_2)\in\RR^2: x_1-x_2 = 2j+1\}\big) = 4-|2j+1|, j = -n,\cdots, n-1.
\]
By \cite[Lemma 3.1]{HanJia:MCOM:2000},
%\cite[Lemma 4.2]{HanJia:SINUM:1999} (see also, \cite[Lemma 3.1]{HanJia:MCOM:2000}).
The matrices $( k^\mu)_{ k\in \Lambda_\epsilon,\; \mu\in\Gamma_{n}}$,$\epsilon=0,1$
are non-singular. Consequently, $\lpf$ must be unique.

Item (i) follows from Proposition~\ref{prop:u:lpm}.
For item (ii), notice that $\wh{\lpf_{2n,2n}^{2D}}(\gbo) = \wh{\mathring{\lpf}}(\gbo)e^{-i\mathbf{c}\cdot \gbo}$, where
\begin{equation}
\label{def:atilde}
\wh{\mathring{\lpf}}(\gbo)
:=\wh{\lpf_{2n}^I}\left(\frac{\go_1+\go_2}{2}\right)+\wh{\lpf_{2n}^I}\left(\frac{\go_1-\go_2}{2}\right)-1.
\end{equation}
One can easily show that $\wh{\mathring{\lpf}}$ satisfies ${\wh{\mathring{\lpf}}(E^\top \cdot)} =  \wh{\mathring{\lpf}}$  for all $E\in D_4$ due to the fact that $\lpf_{2n}^I$ satisfies ${\wh{\lpf_{2n}^I}(-\go)}=\wh{\lpf_{2n}^I}(\go)$ for $\go\in\RR$.  Consequently,
\[
{\wh{\lpf_{2n,2n}^{2D}}(E^\top\gbo)}=\wh{\mathring{\lpf}}(E^\top\gbo) e^{-i{\mathbf{c}}\cdot E^\top \gbo} = \wh{\lpf_{2n,2n}^{2D}}(\gbo)e^{i{\mathbf{c}}\cdot (I_2-E^\top)\gbo}, \;\gbo\in \RR^2,
\]
which is equivalent to \eqref{def:mask-sym-group-time},
i.e., $\lpf^{2D}_{2n,2n}$ is  $D_4$-symmetric about ${\mathbf{c}}=(1/2,1/2)$.

Item (iii) is a direct consequence of \cite[Proposition~2.1]{Han:acha:2004} (also see \cite[Theorem~2.3]{Han:laa:2002}).
In fact, by $N_{\sqrt{2}}=EM_{\sqrt{2}}$ with $E=\left[ \begin{matrix} 0 &1\\ 1 &0\end{matrix}\right]$, $N_{\sqrt{2}}$ is $D_4$-equivalent to $M_{\sqrt{2}}$. Thus, by \cite[Proposition~2.1]{Han:acha:2004},
$\phi^{N_{\sqrt{2}}}=\phi^{M_{\sqrt{2}}}(\cdot+\mathring{\mathbf{c}})$, where $\mathring{\mathbf{c}}:=(M_{\sqrt{2}}-I_2)^{-1}\mathbf{c}-(N_{\sqrt{2}}-I_2)^{-1}\mathbf{c}
=(1,1)$. \eqref{phi:sym} follows from \cite[Proposition~2.1]{Han:acha:2004}.
\end{proof}

\begin{proof}[Proof of Theorem~\ref{thm:sm}]
By the definition of sum rules and $M_{\sqrt{2}}\Z^2=N_{\sqrt{2}}\Z^2$,
it is straightforward to check that $\sr(a,M_{\sqrt{2}})=\sr(a,N_{\sqrt{2}})=\sr(a,2)$.
We now prove $\sm_p(a,M_{\sqrt{2}})=\sm_p(a,2)$. Let $M=M_{\sqrt{2}}$.
By the definition of the subdivision operator in \eqref{sd:op}, we have
\begin{equation}\label{sd}
\wh{\mathcal{S}_{a,M}^nv}(\gbo)=|\det(M)|^n \wh{v}((M^\tp)^{n}\gbo) \wh{a}(\xi)\cdots\wh{a}((M^{\tp})^{n-1}\gbo).
\end{equation}
In particular, noting that $M^2=2I_2$, we have
\[
\wh{\mathcal{S}_{a,M}^n\delta}(\gbo)=2^n \wh{a}(\gbo)\cdots\wh{a}((M^\tp)^{n-1}\gbo)=
\wh{\mathcal{S}_{a,2}^{n_1}\delta}(\go_1) \wh{\mathcal{S}_{a,2}^{n_2}\delta}(\go_1+\go_2),
\]
where
\begin{equation}\label{n12}
n_1:=\lfloor\frac{n+1}{2}\rfloor, \qquad n_2:=n-n_1.
\end{equation}
Therefore, for $\mu_1,\mu_2\in \N_0:=\N\cup\{0\}$, we deduce from the above identity that
\[
[\nabla^{\mu_1}_{e_1} \nabla^{\mu_2}_{e_1+e_2} \mathcal{S}_{a,M}^n \delta](j,k)=
[\nabla^{\mu_1} \mathcal{S}_{a,2}^{n_1}\delta](j-k)[\nabla^{\mu_2} \mathcal{S}_{a,2}^{n_2}\delta](k),\qquad j,k\in \Z,
\]
from which we have
\begin{equation}\label{id}
\|\nabla^{\mu_1}_{e_1} \nabla^{\mu_2}_{e_1+e_2} \mathcal{S}_{a,M}^n \delta\|_{l_p(\Z^2)}
=\| \nabla^{\mu_1} \mathcal{S}_{a,2}^{n_1}\delta\|_{l_p(\Z)} \| \nabla^{\mu_2} \mathcal{S}_{a,2}^{n_2}\delta\|_{l_p(\Z)},\qquad \mu_1,\mu_2, n\in \N_0,
\end{equation}
where $n_1$ and $n_2$ are defined in \eqref{n12}.
Let $m:=\sr(a,2)$. By $\wh{a}(0)=1$, it is known in \cite{Han:jat:2003} and \cite[Theorem~3.1]{Han:simaa:2003} that $\rho_j(a,2)_p\ge 2^{1/p-j}$ for all $j\in \N_0$ and
\begin{equation}\label{rhoj}
\rho_j(a,2)_p=\max(2^{1/p-j},\rho_m(a,2)_p),\qquad j=0,\ldots, m.
\end{equation}
Taking $\mu_1=m$ and $\mu_2=0$ in \eqref{id}, by $\rho_0(a,2)_p\ge 2^{1/p}>0$ and $\lim_{n\to \infty}n_1/n=1/2=\lim_{n\to\infty} n_2/n$, we have
\[
\sqrt{\rho_m(a,2)_p} \sqrt{\rho_0(a,2)_p}=
\lim_{n\to \infty} \| \nabla^m \mathcal{S}_{a,2}^{n_1}\delta\|_{l_p(\Z)}^{1/n}
\lim_{n\to \infty} \| \mathcal{S}_{a,2}^{n_2}\delta\|_{l_p(\Z)}^{1/n}=
\lim_{n\to \infty} \|\nabla^{m}_{e_1} \mathcal{S}_{a,M}^n \delta\|_{l_p(\Z^2)}^{1/n}
\le \rho_m(a,M)_p.
\]
Since $\rho_0(a,2)_p\ge 2^{1/p}$, we conclude from the above inequality that
$\rho_m(a,2)_p \le 2^{-1/p} (\rho_m(a,M)_p)^2$. Consequently, by $|\det(M)|=2$ and $\sr(a,M)=m$,
we have
\[
\sm_p(a,2)=\tfrac{1}{p}-\log_2 \rho_m(a,2)_p\ge \tfrac{1}{p}-\log_2 [2^{-1/p} (\rho_m(a,M)_p)^2]
=\tfrac{2}{p}-2\log_2 \rho_m(a,M)_p=\sm_p(a,M).
\]
This proves $\sm_p(a,2)\ge \sm_p(a,M)$. Conversely, taking $\mu_1=j$ and $\mu_2=m-j$ in \eqref{id} with $0\le j\le m$, we have
\begin{equation}\label{est2}
\lim_{n\to \infty} \|\nabla^{j}_{e_1} \nabla^{m-j}_{e_1+e_2} \mathcal{S}_{a,M}^n \delta\|_{l_p(\Z^2)}^{2/n}
=\lim_{n\to \infty}\| \nabla^{j} \mathcal{S}_{a,2}^{n_1}\delta\|_{l_p(\Z)}^{2/n}
\lim_{n\to \infty} \| \nabla^{m-j} \mathcal{S}_{a,2}^{n_2}\delta\|_{l_p(\Z)}^{2/n}
\le \rho_{j}(a,2)_p\rho_{m-j}(a,2)_p.
\end{equation}
By $\nabla_{e_2} \delta=\nabla_{e_1+e_2}\delta-[\nabla_{e_1} \delta](\cdot-e_2)$,
we see that all $\nabla^{\mu_1}_{e_1}\nabla_{e_2}^{m-\mu_1}\delta$ with $\mu_1=0,\ldots,m$ are finitely linear combinations of $[\nabla^{j}_{e_1}\nabla_{e_1+e_2}^{m-j}\delta](\cdot-k), j=0,\ldots,m$ and $k\in \Z^2$.
If we can prove
\begin{equation}\label{smj}
\rho_{j}(a,2)_p\rho_{m-j}(a,2)_p \le
2^{1/p}\rho_m(a,2)_p,\qquad \forall\, j=0,\ldots, m,
\end{equation}
then it follows from \eqref{est2} that $(\rho_m(a,M)_p)^2\le 2^{1/p}\rho_m(a,2)_p$. Since $m=\sr(a,M)$ and $|\det(M)|=2$,
\[
\sm_p(a,M)=\tfrac{2}{p}-2\log_2 \rho_m(a,M)_p \ge
\tfrac{2}{p}-2\log_2\sqrt{2^{1/p} \rho_m(a,2)_p}=\tfrac{1}{p}-\log_2 \rho_m(a,2)_p=\sm_p(a,2).
\]
Hence, $\sm_p(a,M)\ge \sm_p(a,2)$ and this completes the proof of item (i).

We now prove \eqref{smj}. According to \eqref{rhoj}, we have four cases to consider.
If $\rho_j(a,2)_p=2^{1/p-j}$ and $\rho_{m-j}(a,2)_p=2^{1/p-(m-j)}$, then \eqref{smj} holds, since
\[
\rho_{j}(a,2)_p\rho_{m-j}(a,2)_p=2^{1/p-j} 2^{1/p-(m-j)}=2^{2/p-m}=2^{1/p} 2^{1/p-m}\le
2^{1/p} \rho_m(a,2)_p,
\]
where we used the fact that $\rho_m(a,2)_p\ge 2^{1/p-m}$.
If $\rho_j(a,2)_p=\rho_m(a,2)_p$ and $\rho_{m-j}(a,2)_p=2^{1/p-(m-j)}$, then \eqref{smj} holds, since $\rho_{j}(a,2)_p\rho_{m-j}(a,2)_p=2^{1/p-(m-j)} \rho_m(a,2)_p\le 2^{1/p} \rho_m(a,2)_p$.
The case $\rho_j(a,2)_p=2^{1/p-j}$ and $\rho_{m-j}(a,2)_p=\rho_m(a,2)_p$ is similar.

If $\rho_j(a,2)_p=\rho_m(a,2)_p$ and $\rho_{m-j}(a,2)_p=\rho_m(a,2)_p$, then
\[
\rho_{j}(a,2)_p\rho_{m-j}(a,2)_p=\rho_m(a,2)_p \rho_m(a,2)_p\le
2^{1/p} \rho_m(a,2)_p,
\]
where we used the inequality $\rho_m(a,2)_p\le 2^{1/p}$ which is guaranteed by our assumption $\sm(a,2)_p\ge 0$.
Therefore, \eqref{smj} is verified and this completes the proof of item (i).

We now prove item (ii). The claim $\sr(u*v,M)\ge \sr(u,M)+\sr(v,M)$ can be directly verified by using the definition of sum rules.
By \eqref{sd} and $\wh{u*v}(\gbo)=\wh{u}(\gbo)\wh{v}(\gbo)$, for $\mu,\nu\in \N_0^d$, we have
\[
\nabla^{\mu+\nu} \mathcal{S}_{u*v,M}^n\delta=|\det(M)|^{-n} [\nabla^\mu \mathcal{S}_{u,M}^n\delta]*[\nabla^\nu \mathcal{S}_{v,M}^n\delta].
\]
Consequently, by Cauchy-Schwarz inequality, we have
\[
\|\nabla^{\mu+\nu} \mathcal{S}_{u*v,M}^n\delta\|_{l_\infty(\Z^d)}\le |\det(M)|^{-n}
\|\nabla^\mu \mathcal{S}_{u,M}^n\delta\|_{l_2(\Z^d)} \| \nabla^\nu \mathcal{S}_{v,M}^n\delta\|_{l_2(\Z^d)}
\]
Let $m_1:=\sr(u,M)$ and $m_2:=\sr(v,M)$. Taking $\mu,\nu\in \N_0^d$ with $|\mu|=m_1$ and $|\nu|=m_2$ in the above inequality, we have
\begin{align*}
\lim_{n\to \infty} \|\nabla^{\mu+\nu} \mathcal{S}_{u*v,M}^n\delta\|_{l_\infty(\Z^d)}^{1/n}
&\le |\det(M)|^{-1} \lim_{n\to \infty}
\|\nabla^\mu \mathcal{S}_{u,M}^n\delta\|_{l_2(\Z^d)}^{1/n}
\lim_{n\to \infty} \| \nabla^\nu \mathcal{S}_{v,M}^n\delta\|_{l_2(\Z^d)}^{1/n}\\
&\le |\det(M)|^{-1} \rho_{m_1}(u,M)_2 \rho_{m_2}(v,M)_2.
\end{align*}
Note that any element $\eta\in \N_0^d$ with $|\eta|=m_1+m_2$ can be written as $\eta=\mu+\nu$ with $|\mu|=m_1$ and $|\nu|=m_2$ for some $\mu,\nu\in \N_0^d$. Thus, we deduce from the above inequality that $\rho_{m_1+m_2}(u*v,M)_\infty\le |\det(M)|^{-1} \rho_{m_1}(u,M)_2 \rho_{m_2}(v,M)_2$. Let $m:=\sr(u*v,M)$. By $m\ge m_1+m_2$, we have
\[
\rho_m(u*v,M)_\infty\le  \rho_{m_1+m_2}(u*v,M)_\infty\le |\det(M)|^{-1} \rho_{m_1}(u,M)_2 \rho_{m_2}(v,M)_2,
\]
from which we have
\begin{align*}
\sm_\infty(u*v,M)&=-d\log_{|\det(M)|} \rho_m(u*v,M)_\infty\ge
-d\log_{|\det(M)|} [|\det(M)|^{-1} \rho_{m_1}(u,M)_2 \rho_{m_2}(v,M)_2]\\
&=\tfrac{d}{2}-d\log_{|\det(M)|} \rho_{m_1}(u,M)_2+\tfrac{d}{2}-d\log_{|\det(M)|} \rho_{m_2}(v,M)_2=\sm_2(u,M)+\sm_2(v,M).
\end{align*}
The proof of item (ii) is completed by noting that $\sm_\infty(u*v,M)\le \sm_2(u*v,M)$ always holds.

To prove item (iii), we define $\tilde{u}(k,j):=u(k)\delta(j)$
and $\tilde{v}(j,k):=v(k)\delta(j)$ for all $j, k\in \Z$. That is,
$\tilde{u}$ is the 2D filter by identifying $u$ on $\Z$ with $\Z\times\{0\}$, while
$\tilde{v}$ is the 2D filter by identifying $v$ on $\Z$ with $\{0\}\times \Z$. Since $\sm(u,2)\ge 0$ and $\sm(v,2)\ge 0$, by item (i), we have $\sr(\tilde{u},M_{\sqrt{2}})=\sr(u,2)$, $\sr(\tilde{v},M_{\sqrt{2}})=\sr(v,2)$
and $\sm(\tilde{u},M_{\sqrt{2}})=\sm(u,2)$,
$\sm(\tilde{v},M_{\sqrt{2}})=\sm(v,2)$.
Note that $u\otimes v=\tilde{u}*\tilde{v}$. Now the claim in item (iii) follows directly from item (ii).
\end{proof}

\end{document}